\documentclass[11pt,letterpaper]{article}
\usepackage{times,mathptmx}
\DeclareMathAlphabet{\mathcal}{OMS}{cmsy}{m}{n}
\usepackage{fullpage}
\usepackage{amsmath,amsfonts,amssymb,amsthm,amscd}
\usepackage{tabularx,arydshln}
\usepackage{hyperref,cite}
\newcommand{\vect}[1]{\mathbf{#1}}

\newcommand{\vs}{\vspace{1.5mm}}
\newtheorem{theorem}{Theorem}[section]

\newtheorem{lemma}[theorem]{Lemma}

\newtheorem{definition}[theorem]{Definition}

\newtheorem{assumption}[theorem]{Assumption}
\newcommand{\G}{\mathbb{G}}
\newcommand{\Z}{\mathbb{Z}}
\newcommand{\bits}{\{0,1\}}
\newcommand{\Adv}{\textbf{Adv}}
\newcommand{\mc}[1]{\mathcal{#1}}
\newcommand{\tb}[1]{\textbf{#1}}
\newcommand{\lb}{\linebreak[0]}
\newcommand{\db}{\displaybreak[0]}

\title{Sequential Aggregate Signatures with Short Public Keys\\ without
    Random Oracles\footnote{This paper is the combined and extended version
    of two papers \cite{LeeLY13s,LeeLY13m} that appeared in PKC 2013 and
    ACNS 2013, respectively.}}

\author{
    Kwangsu Lee\footnote{Korea University, Seoul, Korea.
        Email: \texttt{guspin@korea.ac.kr}.}
        \and
    Dong~Hoon Lee\footnote{Korea University, Seoul, Korea.
        Email: \texttt{donghlee@korea.ac.kr}.}
        \and
    Moti Yung\footnote{Google Inc. and Columbia University, NY, USA.
        Email: \texttt{moti@cs.columbia.edu}.}
}

\date{}

\begin{document}

\maketitle

\begin{abstract}
The notion of aggregate signature has been motivated by applications and it
enables any user to compress different signatures signed by different
signers on different messages into a short signature. Sequential aggregate
signature, in turn, is a special kind of aggregate signature that only
allows a signer to add his signature into an aggregate signature in
sequential order. This latter scheme has applications in diversified
settings such as in reducing bandwidth of certificate chains and in secure
routing protocols. Lu, Ostrovsky, Sahai, Shacham, and Waters (EUROCRYPT
2006) presented the first sequential aggregate signature scheme in the
standard model. The size of their public key, however, is quite large
(i.e., the number of group elements is proportional to the security
parameter), and therefore, they suggested as an open problem the
construction of such a scheme with short keys.

In this paper, we propose the first sequential aggregate signature schemes
with short public keys (i.e., a constant number of group elements) in prime
order (asymmetric) bilinear groups that are secure under static assumptions
in the standard model. Furthermore, our schemes employ a constant number of
pairing operations per message signing and message verification operation.
Technically, we start with a public-key signature scheme based on the
recent dual system encryption technique of Lewko and Waters (TCC 2010).
This technique cannot directly provide an aggregate signature scheme since,
as we observed, additional elements should be published in a public key to
support aggregation. Thus, our constructions are careful augmentation
techniques for the dual system technique to allow it to support sequential
aggregate signature schemes. We also propose a multi-signature scheme with
short public parameters in the standard model.
\end{abstract}

\vs \noindent {\bf Keywords:} Public-key signature, Aggregate signature,
Sequential aggregate signature, Dual system encryption, Bilinear maps.

\newpage
\tableofcontents
\newpage

\section{Introduction}

Aggregate signature is a relatively new type of public-key signature (PKS)
that enables any user to combine $n$ signatures signed by $n$ different
signers on $n$ different messages into a short signature. The concept of
public-key aggregate signature (PKAS) was introduced by Boneh, Gentry, Lynn,
and Shacham \cite{BonehGLS03}, and they proposed an efficient PKAS scheme in
the random oracle model using bilinear groups. After that, numerous PKAS
schemes were proposed using bilinear groups \cite{GentryR06,LuOSSW06,
LuOSSW13,BoldyrevaGOY07,BoldyrevaGOY10,AhnGH10,GerbushLOW12} or using
trapdoor permutations \cite{LysyanskayaMRS04,BellareNN07,Neven08}.

One application of aggregate signature is the certificate chains of the
public-key infrastructure (PKI) \cite{BonehGLS03}. The PKI system has a tree
structure, and a certificate for a user consists of a certificate chain from
a root node to a leaf node, each node in the chain signing its predecessor.
If the signatures in the certificate chain are replaced with a single
aggregate signature, then bandwidth for signature transfer can be
significantly saved. Another application is to the secure routing protocol of
the Internet protocol \cite{BonehGLS03}. If each router that participates in
the routing protocol uses a PKAS scheme instead of a PKS scheme, then the
communication overload of signature transfer can be dramatically reduced.
Furthermore, aggregate signatures have other applications such as reducing
bandwidth in sensor networks or ad-hoc networks, as well as in software
authentication in the presence of software update \cite{AhnGH10}.

\subsection{Previous Methods}

Aggregate signature schemes are categorized as {\it full} aggregate
signature, {\it synchronized} aggregate signature, and {\it sequential}
aggregate signature depending on the type of signature aggregation. They have
also been applied to regular signatures in the PKI model and to ID-based
signatures (with a trusted key server).

The first type of aggregate signature is \textit{full aggregate signature},
which enables any user to freely aggregate different signatures of different
signers. This full aggregate signature is the most flexible aggregate
signature since it does not require any restriction on the aggregation step
(though restriction may be needed at times for certain applications).
However, there is only one full aggregate signature scheme, proposed by Boneh
et al. \cite{BonehGLS03}\footnote{Subsequent to our work, Hohenberger et al.
\cite{HohenbergerSW13} proposed an identity-based aggregate signature scheme
that supports full aggregation based on the recently introduced candidate
multilinear maps of Garg et al. \cite{GargGH13}.}. Since this scheme is based
on the short signature scheme of Boneh et al. \cite{BonehLS01}, the signature
length it provides is also very short. However, the security of the scheme
has just been proven in the idealized random oracle model and the number of
pairing operations in the aggregate signature verification algorithm is
proportional to the number of signers in the aggregate signature.

The second type of aggregate signature is \textit{synchronized aggregate
signature}, which enables any user to combine different signatures with the
same synchronizing information into a single signature. The synchronized
aggregate signature has one drawback: all signers should share the same
synchronizing information (such as a time clock or another shared value).
Gentry and Ramzan \cite{GentryR06} introduced the concept of synchronized
aggregate signature. They proposed an identity-based synchronized aggregate
signature scheme using bilinear groups, and they proved its security in the
random oracle model. We note that identity-based aggregate signature (IBAS)
is an ID-based scheme and thus relies on a trusted server knowing all private
keys (i.e., its trust structure is different from that in regular PKI).
However, it also has a notable advantage in that it is not required to
retrieve the public keys of signers in the verification algorithm since an
identity string plays the role of a public key (this lack of public key is
indicated in our comparison table as public key of no size!). Recently, Ahn
et al. \cite{AhnGH10} presented a public-key synchronized aggregate signature
scheme without relying on random oracles.

The third type of aggregate signature is \textit{sequential aggregate
signature} (SAS), which enables each signer to aggregate his signature to a
previously aggregated signature in a sequential order. The sequential
aggregate signature has the obvious limitation of signers being ordered to
aggregate their signatures in contrast to the full aggregate signature and
the synchronized aggregate signature. However, it has an advantage in that it
is not required to share synchronized information among signers in contrast
to the synchronized aggregate signature, and many natural applications lead
themselves to this setting. The concept of sequential aggregate signature was
introduced by Lysyanskaya, Micali, Reyzin, and Shacham
\cite{LysyanskayaMRS04}, and they proposed a public-key sequential aggregate
signature scheme using the certified trapdoor permutations in the random
oracle model. Boldyreva et al. \cite{BoldyrevaGOY07} presented an
identity-based sequential aggregate signature scheme in the random oracle
model using an interactive assumption, but it was shown by Hwang et al.
\cite{HwangLY09} that their construction is not secure. After that, Boldyreva
et al. \cite{BoldyrevaGOY10} proposed a new identity-based sequential
aggregate signature by modifying their previous construction and proved its
security in the generic group model. Recently, Gerbush et al.
\cite{GerbushLOW12} showed that the modified IBAS scheme of Boldyreva et al.
\cite{BoldyrevaGOY10} is secure under static assumptions using the dual form
signatures framework.

The first sequential aggregate signature scheme without random oracle
idealization was proposed by Lu et al. \cite{LuOSSW06,LuOSSW13}. They
converted the PKS scheme of Waters \cite{Waters05} to the PKAS scheme and
proved its security under the well known CDH assumption. However, their
scheme has a drawback since the number of group elements in a public key is
proportional to the security parameter (for a security of $2^{80}$ they need
$160$ elements, or about $80$ elements in a larger group); so they left as an
open question how to design a scheme with shorter public keys. Schr{\"o}der
proposed a PKAS scheme with short public keys relying on the
Camenisch-Lysyanskaya signature scheme \cite{Schroder11}; however the
scheme's security is proven under an interactive assumption (which,
typically, is a relaxation used when designs based on static assumptions are
hard to find).\footnote{Gerbush et al. \cite{GerbushLOW12} showed that a
modified Camenisch-Lysyanskaya signature scheme in composite order groups is
secure under static assumptions. However, it is unclear whether the
construction of Schr{\"o}der can be directly applied to this modified
Camenisch-Lysyanskaya signature scheme. The reason is that aggregating
$\G_{p_1}$ and $\G_{p_2}$ subgroups is hard and a private key element
$g_{2,3} \in \G_{p_2 p_3}$ cannot be generated by the key generation
algorithm of an aggregate signature scheme. Additionally, our work and
findings are independent of the work of Gerbush et al.} Therefore, the
construction of an SAS scheme with short public keys without relaxations such
as random oracles or interactive assumptions was left as an open question.

\begin{table*}[t]
\caption{Comparison of aggregate signature schemes} \label{tab:pkas-compare}
\vs \small \addtolength{\tabcolsep}{0.6pt}
\renewcommand{\arraystretch}{1.2}
    \begin{tabularx}{6.50in}{lcccccccl}
    \hline
    Scheme                  & Type & ROM & KOSK & PK Size & AS Size
                            & Sign Time & Verify Time & Assumption \\
    \hline
    BGLS \cite{BonehGLS03}  & Full & Yes & No & $1k_p$ & $1k_p$
                            & $1$E & $l$P& CDH \\
    \cdashline{2-9}
    GR \cite{GentryR06}     & IB, Sync & Yes & No & -- & $2k_p + \lambda$
                            & $3$E & $3$P + $l$E & CDH \\
    AGH \cite{AhnGH10}      & Sync & Yes & Yes & $1k_p$ & $2k_p + 32$
                            & $6$E & $4$P + $l$E & CDH \\
    AGH \cite{AhnGH10}      & Sync & No  & Yes & $1k_p$ & $2k_p + 32$
                            & $10$E & $8$P + $l$E & CDH \\
    \cdashline{2-9}
    LMRS \cite{LysyanskayaMRS04} & Seq & Yes & No & $1k_f$ & $1k_f$
                            & $l$E & $l$E & cert TDP \\
    Neven \cite{Neven08}    & Seq & Yes & No & $1k_f$ & $1k_f + 2\lambda$
                            & $1$E + $2l$M & $2l$M & uncert CFP \\
    BGOY \cite{BoldyrevaGOY10} & IB, Seq  & Yes & No & -- & $3k_p$
                            & $4$P + $l$E & $4$P + $l$E & Interactive \\
    GLOW \cite{GerbushLOW12} & IB, Seq  & Yes & No & -- & $5k_f$
                            & $10$P + $2l$E & $10$P + $2l$E & Static \\
    \cdashline{3-9}
    LOSSW \cite{LuOSSW06}   & Seq & No & Yes & $2\lambda k_p$ & $2k_p$
                            & $2$P + $4\lambda l$M & $2$P + $2\lambda l$M & CDH \\
    Schr\"oder \cite{Schroder11} & Seq  & No & Yes & $2k_p$ & $4k_p$
                            & $l$P + $2l$E & $l$P + $l$E & Interactive \\
    Ours                    & Seq  & No & Yes & $11k_p$ & $8k_p$
                            & $8$P + $5l$E & $8$P + $4l$E & Static \\
    Ours                    & Seq  & No & Yes & $13k_p$ & $6k_p$
                            & $6$P + $6l$E & $6$P + $3l$E & Static \\
    \hline
    \multicolumn{9}{l}{ROM = random oracle model, KOSK = certified-key model,
    IB = identity based} \\
    \multicolumn{9}{l}{$\lambda$ = security parameter,
    $k_p, k_f$ = the bit size of elements for pairing and factoring, $l$ =
    the number of signers} \\
    \multicolumn{9}{l}{P = pairing computation, E = exponentiation,
    M = multiplication}
    \end{tabularx}
\end{table*}

\subsection{Our Contributions}

Challenged by the above question, the motivation of our research is to
construct an efficient SAS scheme secure in the standard model (i.e., without
employing assumptions such as random oracle or interactive assumptions as
part of the proof) with short public keys (e.g., a constant number of group
elements). To achieve this goal, we use the PKS scheme derived from the
identity-based encryption (IBE) scheme, which adopts the innovative dual
system encryption techniques of Waters \cite{Waters09,LewkoW10}. That is, an
IBE scheme is first converted to a PKS scheme by the clever observation of
Naor \cite{BonehF01}. The PKS schemes that adopt the dual system encryption
techniques are the scheme of Waters \cite{Waters09}, which includes a random
tag in a signature, and the scheme of Lewko and Waters \cite{LewkoW10}, which
does not include a random tag in a signature. The scheme of Waters is not
appropriate to aggregate signatures since the random tags in signatures
cannot be compressed into a single value. The scheme of Lewko and Waters in
composite order groups is easily converted to an aggregate signature scheme
if an element in $\G_{p_3}$ is moved from a private key to a public key, but
it is inefficient because of composite order groups.\footnote{We can safely
move the element in $\G_{p_3}$ from a private key to a public key since it is
always given in assumptions. Lewko obtained a prime order IBE scheme by
translating the Lewko-Waters composite order IBE scheme using the dual
pairing vector spaces \cite{Lewko12}. One may consider to construct an
aggregate signature scheme using this IBE scheme. However, it is not easy to
aggregate individual signatures since the dual orthonormal basis vectors of
each users are randomly generated.}

Therefore, we start the construction from the IBE scheme of Lewko and Waters
(LW-IBE) \cite{LewkoW10} in the prime order (asymmetric) bilinear groups.
However, this LW-PKS scheme, which is directly derived from the LW-IBE
scheme, is not easily converted to an SAS scheme (as far as we see). The
reason is that we need a PKS scheme that supports multi-users and public
re-randomization to construct an SAS scheme by using the randomness reuse
technique of Lu et al. \cite{LuOSSW06}, but the LW-PKS scheme does not
support these two properties. Technically speaking, this directly converted
LW-PKS scheme does not support multi-users and public re-randomization since
group elements $g, u, h \in \G$ cannot be published in a public key. To
resolve this problem, we devised two independent solutions. Our first
solution for this problem is to randomize the verification algorithm of the
LW-PKS scheme and publish $g, u, h \in \G$ in the public key. That is, the
verification components are additionally multiplied by $\hat{v},
\hat{v}^{\nu_3}, \hat{v}^{-\pi}$ to prevent the verification of invalid
signatures. Our second solution for this problem is to randomize the group
elements of the public key. That it, we publish $g w_1^{c_g}, u w_1^{c_u}, h
w_1^{c_h} \in \G$ in the public key instead of $g, u, h \in \G$.

Here we first construct two PKS schemes in prime order (asymmetric) bilinear
groups that support multi-users and public re-randomization by applying our
two solutions to the LW-PKS scheme, and we prove their security by using the
dual system encryption technique. Next, we convert the modified PKS schemes
to SAS schemes with short public keys by using the randomness reuse
technique, and then we prove their security based on the traditional static
assumptions without random oracles. Additionally, we present an efficient
multi-signature scheme based on our modified PKS scheme. Table
\ref{tab:pkas-compare} gives the comparison of past aggregate signature
schemes with ours.

\subsection{Additional Related Work}

There are some works on aggregate signature schemes that allow signers to
communicate with each other or schemes that compress only partial elements of
a signature in the aggregate algorithm \cite{BellareN07,BagherzandiJ10,
Herranz06,BrogleGR12}. Generally, communication resources of computer systems
are very expensive compared with computation resources. Thus, it is preferred
to perform several expensive computational operations rather than one single
communication exchange. Additionally, a signature scheme with added
communications does not correspond to a pure PKS scheme, but corresponds more
to a multi-party protocol. In addition, PKS schemes that compress just
partial elements of signatures cannot be considered aggregate signature
schemes since the total size of signatures is still proportional to the
number of signers.

Another research area related to aggregate signature is multi-signature
\cite{ItakuraN83,Boldyreva03,LuOSSW06}. Multi-signature is a special type of
aggregate signature in which all signers generate signatures on the same
message, and then any user can combine these signatures into a single
signature. Aggregate message authentication code (AMAC) is the symmetric key
analogue of aggregate signature: Katz and Lindell \cite{KatzL08} introduced
the concept of AMAC and showed that it is possible to construct an AMAC
scheme based on any message authentication code scheme.

\section{Preliminaries}

In this section, we define asymmetric bilinear groups and introduce the
complexity assumptions for our schemes. The description of LW-IBE and LW-PKS
schemes is given in Appendix \ref{sec:lw-ibe}.

\subsection{Asymmetric Bilinear Groups}

Let $\G, \hat{\G}$ and $\G_{T}$ be multiplicative cyclic groups of prime
order $p$. Let $g$ and $\hat{g}$ be generators of $\G$ and $\hat{\G}$,
respectively. The bilinear map $e : \G \times \hat{\G} \rightarrow \G_{T}$
has the following properties:
\begin{enumerate}
\item Bilinearity: $\forall u \in \G, \forall \hat{v} \in \hat{\G}$ and
    $\forall a,b \in \Z_p$, $e(u^a,\hat{v}^b) = e(u,\hat{v})^{ab}$.
\item Non-degeneracy: $e(g,\hat{g}) \neq 1$, that is, $e(g,\hat{g})$ is a
    generator of $\G_T$.
\end{enumerate}
We say that $\G, \hat{\G}, \G_T$ are bilinear groups with no efficiently
computable isomorphisms if the group operations in $\G, \hat{\G},$ and $\G_T$
as well as the bilinear map $e$ are all efficiently computable, but there are
no efficiently computable isomorphisms between $\G$ and $\hat{\G}$.

\subsection{Complexity Assumptions}

We employ four assumptions in prime order bilinear groups. The SXDH and DBDH
assumptions have been used extensively, while the LW1 and LW2 assumptions
were introduced by Lewko and Waters \cite{LewkoW10}.

\begin{assumption}[Symmetric eXternal Diffie-Hellman, SXDH]
Let $(p, \G, \hat{\G}, \G_T, e)$ be a description of the asymmetric bilinear
group of prime order $p$. Let $g, \hat{g}$ be generators of $\G, \hat{\G}$
respectively. The assumption is that if the challenge values
    $$D = ((p, \G, \hat{\G}, \G_T, e),
    g, \hat{g}, \hat{g}^a, \hat{g}^b) \mbox{ and } T,$$
are given, no PPT algorithm $\mc{B}$ can distinguish $T = T_0 = \hat{g}^{ab}$
from $T = T_1 = \hat{g}^c$ with more than a negligible advantage. The
advantage of $\mc{B}$ is defined as $\Adv_{\mc{B}}^{SXDH}(\lambda) = \big|
\Pr[\mc{B}(D, T_0)=0] - \Pr[\mc{B}(D, T_1)=0] \big|$ where the probability is
taken over the random choice of $a, b, c \in \Z_p$.
\end{assumption}

\begin{assumption}[LW1]
Let $(p, \G, \hat{\G}, \G_T, e)$ be a description of the asymmetric bilinear
group of prime order $p$ with the security parameter $\lambda$. Let $g,
\hat{g}$ be generators of $\G, \hat{\G}$ respectively. The assumption is that
if the challenge values
    \begin{align*}
    D = ((p, \G, \hat{\G}, \G_T, e),
        g, g^b, \hat{g}, \hat{g}^a, \hat{g}^b, \hat{g}^{ab^2}, \hat{g}^{b^2},
        \hat{g}^{b^3}, \hat{g}^c, \hat{g}^{ac}, \hat{g}^{bc}, \hat{g}^{b^2 c},
        \hat{g}^{b^3 c}) \mbox{ and } T
    \end{align*}
are given, no PPT algorithm $\mc{B}$ can distinguish $T = T_0 = \hat{g}^{ab^2
c}$ from $T = T_1 = \hat{g}^d$ with more than a negligible advantage. The
advantage of $\mc{B}$ is defined as $\Adv_{\mc{B}}^{LW1} (\lambda) = \big|
\Pr[\mc{B}(D,T_0)=0] - \Pr[\mc{B}(D,T_1)=0] \big|$ where the probability is
taken over the random choice of $a, b, c, d \in \Z_p$.
\end{assumption}

\begin{assumption}[LW2]
Let $(p, \G, \hat{\G}, \G_T, e)$ be a description of the asymmetric bilinear
group of prime order $p$. Let $g, \hat{g}$ be generators of $\G, \hat{\G}$
respectively. The assumption is that if the challenge values
    $$D = ((p, \G, \hat{\G}, \G_T, e),
    g, g^a, g^b, g^c, \hat{g}, \hat{g}^a, \hat{g}^{a^2}, \hat{g}^{bx},
    \hat{g}^{abx}, \hat{g}^{a^2x}) \mbox{ and } T,$$
are given, no PPT algorithm $\mc{B}$ can distinguish $T = T_0 = g^{bc}$ from
$T = T_1 = g^d$ with more than a negligible advantage. The advantage of
$\mc{B}$ is defined as $\Adv_{\mc{B}}^{LW2}(\lambda) = \big| \Pr[\mc{B}(D,
T_0)=0] - \Pr[\mc{B}(D, T_1)=0] \big|$ where the probability is taken over
the random choice of $a, b, c, x, d \in \Z_p$.
\end{assumption}

\begin{assumption}[Decisional Bilinear Diffie-Hellman, DBDH]
Let $(p, \G, \hat{\G}, \G_T, e)$ be a description of the asymmetric bilinear
group of prime order $p$. Let $g, \hat{g}$ be generators of $\G, \hat{\G}$
respectively. The assumption is that if the challenge values
    $$D = ((p, \G, \hat{\G}, \G_T, e),
    g, g^a, g^b, g^c, \hat{g}, \hat{g}^a, \hat{g}^b, \hat{g}^c)
    \mbox{ and } T,$$
are given, no PPT algorithm $\mc{B}$ can distinguish $T = T_0 = e(g,
\hat{g})^{abc}$ from $T = T_1 = e(g, \hat{g})^{d}$ with more than a
negligible advantage. The advantage of $\mc{B}$ is defined as
$\Adv_{\mc{B}}^{DBDH}(\lambda) = \big| \Pr[\mc{B}(D, T_0)=0] - \Pr[\mc{B}(D,
T_1)=0] \big|$ where the probability is taken over the random choice of $a,
b, c, d \in \Z_p$.
\end{assumption}

The LW1 and LW2 assumptions are falsifiable since they are not interactive
(or even $q$-type) assumptions and they obviously hold in the generic
bilinear group model since the target polynomial in $T$ is independent of
given polynomials in $D$.

\section{Public-Key Signature} \label{sec:pks}

In this section, we propose two PKS schemes with short public keys and prove
their security under static assumptions.

\subsection{Definitions}

The concept of PKS was introduced by Diffie and Hellman \cite{DiffieH76}. In
PKS, a signer first generates a public key and a private key, and then he
publishes the public key. The signer generates a signature on a message by
using his private key. A verifier can check the validity of the signer's
signature on the message by using the signer's public key. A PKS scheme is
formally defined as follows:

\begin{definition}[Public-Key Signature]
A public key signature (PKS) scheme consists of three PPT algorithms
\tb{KeyGen}, \tb{Sign}, and \tb{Verify}, which are defined as follows:
\begin{description}
\item $\tb{KeyGen}(1^{\lambda})$. The key generation algorithm takes as
    input the security parameters $1^{\lambda}$ and outputs a public key
    $PK$ and a private key $SK$.

\item $\tb{Sign}(M, SK)$. The signing algorithm takes as input a message
    $M$ and a private key $SK$ and outputs a signature $\sigma$.

\item $\tb{Verify}(\sigma, M, PK)$. The verification algorithm takes as
    input a signature $\sigma$, a message $M$, and a public key $PK$ and
    outputs either $1$ or $0$, depending on the validity of the signature.
\end{description}
The correctness requirement is that for any $(PK,SK)$ output by \tb{KeyGen}
and any $M \in \mathcal{M}$, we have $\tb{Verify} \lb (\tb{Sign} (M,SK), M,
PK) = 1$. We can relax this notion to require that the verification is
correct with overwhelming probability over all the randomness of the
experiment.
\end{definition}

The security model of PKS is defined as existential unforgeability under a
chosen message attack (EUF-CMA), and this was formally defined by Goldwasser
et al. \cite{GoldwasserMR88}. In this security model, an adversary adaptively
requests a polynomial number of signatures on messages through the signing
oracle, and he finally outputs a forged signature on a message $M^*$. If the
message $M^*$ was not queried to the signing oracle and the forged signature
is valid, then the adversary wins this game. The security of PKS is formally
defined as follows:

\begin{definition}[Security]
The security notion of existential unforgeability under a chosen message
attack is defined in terms of the following experiment between a challenger
$\mc{C}$ and a PPT adversary $\mc{A}$:
\begin{enumerate}
\item \tb{Setup}: $\mc{C}$ first generates a key pair $(PK,SK)$ by running
    \tb{KeyGen}, and gives $PK$ to $\mc{A}$.

\item \tb{Signature Query}: Then $\mc{A}$, adaptively and polynomially
    many times, requests a signature query on a message $M$ under the
    challenge public key $PK$, and receives a signature $\sigma$
    generated by running \tb{Sign}.

\item \tb{Output}: Finally, $\mc{A}$ outputs a forged signature
    $\sigma^*$ on a message $M^*$. $\mc{C}$ then outputs $1$ if the
    forged signature satisfies the following two conditions, or outputs
    $0$ otherwise: 1) $\tb{Verify}(\sigma^*, M^*, PK) = 1$ and 2) $M^*$
    was not queried by $\mc{A}$ to the signing oracle.
\end{enumerate}
The advantage of $\mc{A}$ is defined as $\Adv_{\mc{A}}^{PKS}(\lambda) = \Pr
[\mc{C} = 1]$ where the probability is taken over all the randomness of the
experiment. A PKS scheme is existentially unforgeable under a chosen message
attack if all PPT adversaries have at most a negligible advantage in the
above experiment (for a large enough security parameter).
\end{definition}

\subsection{Construction} \label{sec:pks-prime}

We construct PKS schemes with a short public key that will be augmented to
support \textit{multi-users} and \textit{public re-randomization}. To
construct a PKS scheme with a short public key, we may convert the LW-IBE
scheme \cite{LewkoW10} in prime order groups to the LW-PKS scheme in prime
order groups by using the transformation of Naor \cite{BonehF01}. However,
this directly converted LW-PKS scheme does not support multi-users and public
re-randomization since it is necessary to publish additional public key
components: Specifically, we need to publish an element $g$ for multi-users
and elements $g, u, h$ for public re-randomization. Note that $\hat{g},
\hat{u}, \hat{h}$ are already in the public key, but $g, u, h$ are not. One
may try to publish $g, u, h$ in the public key, but a technical difficulty
arises in this case in that the simulator of the security proof can easily
distinguish from the normal verification algorithm to the semi-functional
one, without using an adversary. Thus the simulator of Lewko and Waters sets
the CDH value into the elements $g, u, h$ to prevent the simulator from
creating these elements.

To solve this problem, we devise two independent solutions. The first
solution allows a PKS scheme to safely publish elements $g, u, h$ in the
public key for multi-users and public re-randomization. The main idea is to
additionally randomize the verification components using $\hat{v},
\hat{v}^{\nu_3}, \hat{v}^{-\pi}$ in the verification algorithm. If a valid
signature is given in the verification algorithm, then the additionally added
randomization elements $\hat{v}, \hat{v}^{\nu_3}, \hat{v}^{-\pi}$ are
canceled. Otherwise, the added randomization components prevent the
verification of an invalid signature. Therefore, the simulator of the
security proof cannot detect the changes of the verification algorithm even
if $g, u, h$ are published, since the additional elements $\hat{v},
\hat{v}^{\nu_3}, \hat{v}^{-\pi}$ prevent the signature verification.

Our second solution for this problem is to publish randomized components $g
w_1^{c_g}, u w_1^{c_u}, h w_1^{c_h}$ that are additionally multiplied with
random elements rather than directly publishing $g, u, h$. In this case, the
simulator can create these elements since the random exponents $c_g, c_u,
c_h$ can be used to cancel out the CDH value embedded in the elements $g, u,
h$. Additionally, the simulator cannot detect the changes of verification
components for the forged signature because of the added elements $w_1^{c_g},
w_1^{c_u}, w_1^{c_h}$. This solution does not increase the number of group
elements in the signatures, rather it increases the number of public keys
since additional elements $w_2^{c_g}, w^{c_g}, w_2^{c_u}, w^{c_u}, w_2^{c_h},
w^{c_h}$ should be published.

\subsubsection{Our PKS1 Scheme}

Our first PKS scheme in prime order bilinear groups is described as follows:

\begin{description}
\item [\textbf{PKS1.KeyGen}($1^\lambda$):] This algorithm first generates
    the asymmetric bilinear groups $\G, \hat{\G}$ of prime order $p$ of bit
    size $\Theta(\lambda)$. It chooses random elements $g, w \in \G$ and
    $\hat{g}, \hat{v} \in \hat{\G}$. Next, it chooses random exponents
    $\nu_1, \nu_2, \nu_3, \phi_1, \phi_2, \phi_3 \in \Z_p$ and sets $\tau =
    \phi_1 + \nu_1 \phi_2 + \nu_2 \phi_3, \pi = \phi_2 + \nu_3 \phi_3$. It
    selects random exponents $\alpha, x, y \in \Z_p$ and sets $u = g^x, h =
    g^y, \hat{u} = \hat{g}^x, \hat{h} = \hat{g}^y, w_1 = w^{\phi_1}, w_2 =
    w^{\phi_2}, w_3 = w^{\phi_3}$. It outputs a private key $SK = \alpha$
    and a public key as
    \begin{align*}
    PK = \Big(~
    &   (p, \G, \hat{\G}, \G_T, e),~
        g, u, h,~ w_1, w_2, w_3, w,~
        \hat{g}, \hat{g}^{\nu_1}, \hat{g}^{\nu_2}, \hat{g}^{-\tau},~ \\
    &   \hat{u}, \hat{u}^{\nu_1}, \hat{u}^{\nu_2}, \hat{u}^{-\tau},~
        \hat{h}, \hat{h}^{\nu_1}, \hat{h}^{\nu_2}, \hat{h}^{-\tau},~
        \hat{v}, \hat{v}^{\nu_3}, \hat{v}^{-\pi},~
        \Omega = e(g,\hat{g})^{\alpha}
    ~\Big).
    \end{align*}

\item [\textbf{PKS1.Sign}($M, SK$):] This algorithm takes as input a
    message $M \in \bits^k$ where $k < \lambda$ and a private key $SK =
    \alpha$. It selects random exponents $r, c_1, c_2 \in \Z_p$ and outputs
    a signature as
    \begin{align*}
    \sigma = \Big(~
    &   W_{1,1} = g^{\alpha} (u^M h)^r w_1^{c_1},
        W_{1,2} = w_2^{c_1},
        W_{1,3} = w_3^{c_1},
        W_{1,4} = w^{c_1},~ \\
    &   W_{2,1} = g^r w_1^{c_2},
        W_{2,2} = w_2^{c_2},
        W_{2,3} = w_3^{c_2},
        W_{2,4} = w^{c_2}
    ~\Big).
    \end{align*}

\item [\textbf{PKS1.Verify}($\sigma, M, PK$):] This algorithm takes as
    input a signature $\sigma$ on a message $M \in \bits^k$ under a public
    key $PK$. It first chooses random exponents $t, s_1, s_2 \in \Z_p$ and
    computes verification components as
    \begin{align*}
    &   V_{1,1} = \hat{g}^t,
        V_{1,2} = (\hat{g}^{\nu_1})^t \hat{v}^{s_1},
        V_{1,3} = (\hat{g}^{\nu_2})^t (\hat{v}^{\nu_3})^{s_1},
        V_{1,4} = (\hat{g}^{-\tau})^t (\hat{v}^{-\pi})^{s_1}, \\
    &   V_{2,1} = (\hat{u}^M \hat{h})^t,
        V_{2,2} = ((\hat{u}^{\nu_1})^M \hat{h}^{\nu_1})^t \hat{v}^{s_2},
        V_{2,3} = ((\hat{u}^{\nu_2})^M \hat{h}^{\nu_2})^t (\hat{v}^{\nu_3})^{s_2},
        V_{2,4} = ((\hat{u}^{-\tau})^M \hat{h}^{-\tau})^t (\hat{v}^{-\pi})^{s_2}.
    \end{align*}
    Next, it verifies that $\prod_{i=1}^4 e(W_{1,i}, V_{1,i}) \cdot
    \prod_{i=1}^4 e(W_{2,i}, V_{2,i})^{-1} \stackrel{?}{=} \Omega^t$. If
    this equation holds, then it outputs $1$. Otherwise, it outputs $0$.
\end{description}

We note that the inner product of $(\phi_1, \phi_2, \phi_3, 1)$ and $(1,
\nu_1, \nu_2, -\tau)$ is zero since $\tau = \phi_1 + \nu_1 \phi_2 + \nu_2
\phi_3$, and the inner product of $(\phi_1, \phi_2, \phi_3, 1)$ and $(0, 1,
\nu_3, -\pi)$ is zero since $\pi = \phi_2 + \nu_3 \phi_3$. Using these facts,
the correctness of PKS is easily obtained from the equation
    \begin{align*}
    \prod_{i=1}^4 e(W_{1,i}, V_{1,i}) \cdot \prod_{i=1}^4 e(W_{2,i}, V_{2,i})^{-1}
        = e(g^{\alpha} (u^M h)^r, \hat{g}^t) \cdot
          e(g^r, (\hat{u}^M \hat{h})^t)^{-1}
        = \Omega^t.
    \end{align*}

\subsubsection{Our PKS2 Scheme}

Our second PKS scheme in prime order bilinear groups is described as follows:

\begin{description}
\item [\tb{PKS2.KeyGen}($1^\lambda$):] This algorithm first generates the
    asymmetric bilinear groups $\G, \hat{\G}$ of prime order $p$ of bit
    size $\Theta(\lambda)$. It chooses random elements $g, w \in \G$ and
    $\hat{g} \in \hat{\G}$. Next, it selects random exponents $\nu, \phi_1,
    \phi_2 \in \Z_p$ and sets $\tau = \phi_1 + \nu \phi_2$. It also selects
    random exponents $\alpha, x, y \in \Z_p$ and sets $u = g^{x}, h =
    g^{y}, \hat{u} = \hat{g}^{x}, \hat{h} = \hat{g}^y, w_1 = w^{\phi_1},
    w_2 = w^{\phi_2}$. It outputs a private key $SK = (\alpha, g, u, h)$
    and a public key by selecting random values $c_g, c_u, c_h \in \Z_p$ as
    \begin{align*}
    PK = \Big(~
    &   (p, \G, \hat{\G}, \G_T, e),~
        g w_1^{c_g}, w_2^{c_g}, w^{c_g},~
        u w_1^{c_u}, w_2^{c_u}, w^{c_u},~
        h w_1^{c_h}, w_2^{c_h}, w^{c_h},~ \\
    &   w_1, w_2, w,~
        \hat{g}, \hat{g}^{\nu}, \hat{g}^{-\tau},~
        \hat{u}, \hat{u}^{\nu}, \hat{u}^{-\tau},~
        \hat{h}, \hat{h}^{\nu}, \hat{h}^{-\tau},~
        \Omega = e(g,\hat{g})^{\alpha}
    ~\Big).
    \end{align*}

\item [\tb{PKS2.Sign}($M, SK$):] This algorithm takes as input a message $M
    \in \Z_p$ and a private key $SK = (\alpha, g, u, h)$ with $PK$. It
    selects random exponents $r, c_1, c_2 \in \Z_p$ and outputs a signature
    as
    \begin{align*}
    \sigma = \Big(~
    &   W_{1,1} = g^{\alpha} (u^M h)^r w_1^{c_1},~
        W_{1,2} = w_2^{c_1},~
        W_{1,3} = w^{c_1},~ \\
    &   W_{2,1} = g^r w_1^{c_2},~
        W_{2,2} = w_2^{c_2},~
        W_{2,3} = w^{c_2}
    ~\Big).
    \end{align*}

\item [\tb{PKS2.Verify}($\sigma, M, PK$):] This algorithm takes as input a
    signature $\sigma$ on a message $M \in \Z_p$ under a public key $PK$.
    It chooses a random exponent $t \in \Z_p$ and computes verification
    components as
    \begin{align*}
    &   V_{1,1} = \hat{g}^t,~
        V_{1,2} = (\hat{g}^{\nu})^t,~
        V_{1,3} = (\hat{g}^{-\tau})^t,~ \\
    &   V_{2,1} = (\hat{u}^M \hat{h})^t,~
        V_{2,2} = ((\hat{u}^{\nu})^M \hat{h}^{\nu})^t,~
        V_{2,3} = ((\hat{u}^{-\tau})^M \hat{h}^{-\tau})^t.
    \end{align*}
    Next, it verifies that $\prod_{i=1}^3 e(W_{1,i}, V_{1,i}) \cdot
    \prod_{i=1}^3 e(W_{2,i}, V_{2,i})^{-1} \stackrel{?}{=} \Omega^t$. If
    this equation holds, then it outputs $1$. Otherwise, it outputs $0$.
\end{description}

We note that the inner product of $(\phi_1, \phi_2, 1)$ and $(1, \nu, -\tau)$
is zero since $\tau = \phi_1 + \nu \phi_2$. Using this fact, the correctness
of PKS is easily obtained from the following equation
    \begin{align*}
    \prod_{i=1}^3 e(W_{1,i}, V_{1,i}) \cdot \prod_{i=1}^3 e(W_{2,i}, V_{2,i})^{-1}
        = e(g^{\alpha} (u^M h)^r, \hat{g}^t) \cdot
          e(g^r, (\hat{u}^M \hat{h})^t)^{-1}
        = \Omega^t.
    \end{align*}

\subsection{Security Analysis}

We prove the security of our PKS schemes without random oracles under static
assumptions. To prove the security, we use the dual system encryption
technique of Lewko and Waters \cite{LewkoW10}. The dual system encryption
technique was originally developed to prove the full-model security of IBE
and its extensions, but it also can be used to prove the security of PKS by
using the transformation of Naor \cite{BonehF01}. Note that Gerbush et al.
\cite{GerbushLOW12} developed the dual form signature technique that is a
variation of the dual system encryption technique to prove the security of
their PKS schemes.

\subsubsection{Analysis of PKS1}

\begin{theorem} \label{thm:pks1-prime}
The above \tb{PKS1} scheme is existentially unforgeable under a chosen
message attack if the SXDH, LW2, DBDH assumptions hold. That is, for any PPT
adversary $\mc{A}$, there exist PPT algorithms $\mc{B}_1, \mc{B}_2, \mc{B}_3$
such that
    $\Adv_{\mc{A}}^{PKS}(\lambda)
    \leq \Adv_{\mc{B}_1}^{SXDH}(\lambda) + q \Adv_{\mc{B}_2}^{LW2}(\lambda) +
    \Adv_{\mc{B}_3}^{DBDH}(\lambda)$
where $q$ is the maximum number of signature queries of $\mc{A}$.
\end{theorem}

\begin{proof}
To use the dual system encryption technique of Lewko and Waters
\cite{LewkoW10}, we first describe a semi-functional signing algorithm and a
semi-functional verification algorithm. They are not used in a real system;
rather, they are used in the security proof.
When comparing our proof to that of Lewko and Waters, we employ a different
assumption since we have published additional elements $g, u, h$ used in
aggregation (in fact, direct adaptation of the earlier technique will break
the assumption and thus the proof). A crucial idea in our proof is that we
have added elements $\hat{v}, \hat{v}^{\nu_3}, \hat{v}^{-\pi}$ in the public
key that are used in randomization of the verification algorithm. In the
security proof when moving from normal to semi-functional verification, it is
the randomization elements $\hat{v}, \hat{v}^{\nu_3}, \hat{v}^{-\pi}$ that
are expanded to the semi-functional space; this enables deriving
semi-functional verification as part of the security proof under our
assumption, without being affected by the publication of the additional
public key elements used for aggregation.

For the semi-functional signing and verification, we set $f = g^{y_f},
\hat{f} = \hat{g}^{y_f}$ where $y_f$ is a random exponent in $\Z_p$.

\begin{description}
\item [\textbf{PKS1.SignSF}.] The semi-functional signing algorithm first
    creates a normal signature using the private key. Let $(W'_{1,1},
    \ldots, W'_{2,4})$ be the normal signature of a message $M$ with random
    exponents $r, c_1, c_2 \in \Z_p$. It selects random exponents $s_k, z_k
    \in \Z_p$ and outputs a semi-functional signature as
    \begin{align*}
    \sigma = \Big(~
    &   W_{1,1} = W'_{1,1} (f^{\nu_1 \nu_3 - \nu_2})^{s_k z_k},~
        W_{1,2} = W'_{1,2} (f^{-\nu_3})^{s_k z_k},~
        W_{1,3} = W'_{1,3} f^{s_k z_k},~
        W_{1,4} = W'_{1,4},~\\
    &   W_{2,1} = W'_{2,1} (f^{\nu_1 \nu_3 - \nu_2})^{s_k},~
        W_{2,2} = W'_{2,2} (f^{-\nu_3})^{s_k},~
        W_{2,3} = W'_{2,3} f^{s_k},~
        W_{2,4} = W'_{2,4}
    ~\Big).
    \end{align*}

\item [\textbf{PKS1.VerifySF}.] The semi-functional verification algorithm
    first creates normal verification components using the public key. Let
    $(V'_{1,1}, \ldots, V'_{2,4})$ be the normal verification components
    with random exponents $t, s_1, s_2 \in \Z_p$. It chooses random
    exponents $s_c, z_c \in \Z_p$ and computes semi-functional verification
    components as
    \begin{align*}
    &   V_{1,1} = V'_{1,1},~
        V_{1,2} = V'_{1,2},~
        V_{1,3} = V'_{1,3} \hat{f}^{s_c},~
        V_{1,4} = V'_{1,4} (\hat{f}^{-\phi_3})^{s_c}, \\
    &   V_{2,1} = V'_{2,1},~
        V_{2,2} = V'_{2,2},~
        V_{2,3} = V'_{2,3} \hat{f}^{s_c z_c},~
        V_{2,4} = V'_{2,4} (\hat{f}^{-\phi_3})^{s_c z_c}.
    \end{align*}
    Next, it verifies that $\prod_{i=1}^4 e(W_{1,i}, V_{1,i}) \cdot
    \prod_{i=1}^4 e(W_{2,i}, V_{2,i})^{-1} \stackrel{?}{=} \Omega^t$. If
    this equation holds, then it outputs 1. Otherwise, it outputs 0.
\end{description}

\noindent Note that if the semi-functional verification algorithm verifies a
semi-functional signature, then the left part of the above verification
equation contains an additional random element $e(f, \hat{f})^{s_k s_c (z_k -
z_c)}$. If $z_k = z_c$, then the semi-functional verification algorithm
succeeds. In this case, we say that the signature is \textit{nominally}
semi-functional.

The security proof uses a sequence of games $\tb{G}_0, \tb{G}_1, \tb{G}_2$,
and $\tb{G}_3$: The first game $\tb{G}_0$ will be the original security game
and the last game $\tb{G}_3$ will be a game such that an adversary $\mc{A}$
has no advantage. Formally, the hybrid games are defined as follows:

\begin{description}
\item [\textbf{Game} $\tb{G}_0$.] This game is the original security game.
    In this game, the signatures that are given to $\mc{A}$ are normal and
    the challenger use the normal verification algorithm \textbf{Verify} to
    check the validity of the forged signature of $\mc{A}$.

\item [\textbf{Game} $\tb{G}_1$.] We first modify the original game to a
    new game $\tb{G}_1$. This game is almost identical to $\tb{G}_0$ except
    that the challenger uses the semi-functional verification algorithm
    \textbf{VerifySF} to check the validity of the forged signature of
    $\mc{A}$.

\item [\textbf{Game} $\tb{G}_2$.] Next, we change $\tb{G}_1$ to a new game
    $\tb{G}_2$. This game is the same as the $\tb{G}_1$ except that the
    signatures that are given to $\mc{A}$ will be semi-functional. At this
    moment, the signatures are semi-functional and the challenger uses the
    semi-functional verification algorithm \textbf{VerifySF} to check the
    validity of the forged signature. Suppose that $\mc{A}$ makes at most
    $q$ signature queries. For the security proof, we define a sequence of
    hybrid games $\tb{G}_{1,0}, \ldots, \tb{G}_{1,k}, \ldots, \tb{G}_{1,q}$
    where $\tb{G}_{1,0} = \tb{G}_1$. In $\tb{G}_{1,k}$, a normal signature
    is given to $\mc{A}$ for all $j$-th signature queries such that $j >
    k$, and a semi-functional signature is given to $\mc{A}$ for all $j$-th
    signature queries such that $j \leq k$. It is obvious that
    $\tb{G}_{1,q}$ is equal to $\tb{G}_2$.

\item [\textbf{Game} $\tb{G}_3$.] Finally, we define a new game $\tb{G}_3$.
    This game differs from $\tb{G}_2$ in that the challenger always rejects
    the forged signature of $\mc{A}$. Therefore, the advantage of this game
    is zero since $\mc{A}$ cannot win this game.
\end{description}

For the security proof, we show the indistinguishability of each hybrid game.
We informally describe the meaning of each indistinguishability as follows:
\begin{itemize}
\item Indistinguishability of $\tb{G}_0$ and $\tb{G}_1$: This property
    shows that $\mc{A}$ cannot forge a semi-functional signature if it is
    only given normal signatures. That is, if $\mc{A}$ forges a
    semi-functional signature, then it can distinguish $\tb{G}_0$ from
    $\tb{G}_1$.

\item Indistinguishability of $\tb{G}_1$ and $\tb{G}_2$: This property
    shows that the probability of $\mc{A}$ forging a normal signature is
    almost the same when the signatures given to the adversary are changed
    from a normal type to a semi-functional type. That is, if the
    probability of $\mc{A}$ forging a normal signature is different in
    $\tb{G}_1$ and $\tb{G}_2$, then $\mc{A}$ can distinguish the two games.

\item Indistinguishability of $\tb{G}_2$ and $\tb{G}_3$: This property
    shows that $\mc{A}$ cannot forge a normal signature if it is only given
    semi-functional signatures. That is, if $\mc{A}$ forges a normal
    signature, then it can distinguish $\tb{G}_2$ from $\tb{G}_3$.
\end{itemize}

The security (unforgeability) of our PKS scheme follows from a hybrid
argument. We first consider an adversary $\mc{A}$ attacking our PKS scheme in
the original security game $\tb{G}_0$. By the indistinguishability of
$\tb{G}_0$ and $\tb{G}_1$, we have that $\mc{A}$ can forge a normal signature
with a non-negligible $\epsilon$ probability, but it can forge a
semi-functional signature with only a negligible probability. Now we should
show that the $\epsilon$ probability of $\mc{A}$ forging a normal signature
is also negligible. By the indistinguishability of $\tb{G}_1$ and $\tb{G}_2$,
we have that the $\epsilon$ probability of $\mc{A}$ forging a normal
signature is almost the same when the signatures given to $\mc{A}$ are
changed from a normal type to a semi-functional type. Finally, by the
indistinguishability of $\tb{G}_2$ and $\tb{G}_3$, we have that $\mc{A}$ can
forge a normal signature with only a negligible probability. Summing up, we
obtain that the probability of $\mc{A}$ forging a semi-functional signature
is negligible (from the indistinguishability of $\tb{G}_0$ and $\tb{G}_1$)
and the probability of $\mc{A}$ forging a normal signature is also negligible
(from the indistinguishability of $\tb{G}_2$ and $\tb{G}_3$).

Let $\Adv_{\mc{A}}^{G_j}$ be the advantage of $\mc{A}$ in $\tb{G}_j$ for
$j=0, \ldots, 3$. Let $\Adv_{\mc{A}}^{G_{1,k}}$ be the advantage of $\mc{A}$
in $\tb{G}_{1,k}$ for $k=0, \ldots, q$. It is clear that
    $\Adv_{\mc{A}}^{G_0} = \Adv_{\mc{A}}^{PKS}(\lambda)$,
    $\Adv_{\mc{A}}^{G_{1,0}} = \Adv_{\mc{A}}^{G_1}$,
    $\Adv_{\mc{A}}^{G_{1,q}} = \Adv_{\mc{A}}^{G_2}$, and
    $\Adv_{\mc{A}}^{G_3} = 0$.
From the following three Lemmas, we prove that it is hard for $\mc{A}$ to
distinguish $\tb{G}_{i-1}$ from $\tb{G}_{i}$ under the given assumptions.
Therefore, we have that
    \begin{align*}
    \Adv_{\mc{A}}^{PKS}(\lambda)
    & = \Adv_{\mc{A}}^{G_0} +
        \sum_{i=1}^2 \big( \Adv_{\mc{A}}^{G_i} - \Adv_{\mc{A}}^{G_i} \big)
        - \Adv_{\mc{A}}^{G_3}
    \leq \sum_{i=1}^3 \big| \Adv_{\mc{A}}^{G_{i-1}} - \Adv_{\mc{A}}^{G_i} \big| \\
    & = \Adv_{\mc{B}_1}^{SXDH}(\lambda) +
        \sum_{k=1}^q \Adv_{\mc{B}_2}^{LW2}(\lambda) +
        \Adv_{\mc{B}_3}^{DBDH}(\lambda).
    \end{align*}
This completes our proof.
\end{proof}

\begin{lemma} \label{lem:pks1-prime-1}
If the SXDH assumption holds, then no polynomial-time adversary can
distinguish between $\tb{G}_0$ and $\tb{G}_1$ with non-negligible advantage.
That is, for any adversary $\mc{A}$, there exists a PPT algorithm $\mc{B}_1$
such that
    $\big| \Adv_{\mc{A}}^{G_0} - \Adv_{\mc{A}}^{G_1} \big|
    = \Adv_{\mc{B}_1}^{SXDH}(\lambda)$.
\end{lemma}

\begin{proof}
Before proving this lemma, we introduce the parallel-SXDH assumption as
follows: Let $(p, \G, \hat{\G}, \G_T, e)$ be a description of the asymmetric
bilinear group of prime order $p$. Let $k, \hat{k}$ be generators of $\G,
\hat{\G}$ respectively. The assumption is stated as following: given a
challenge tuple
    $D = ((p, \G, \hat{\G}, \G_T, e),
    k, \hat{k}^a, \hat{k}^{d_1}, \hat{k}^{d_2})$ and $T = (A_1, A_2)$,
it is hard to decide whether $T = (\hat{k}^{ad_1}, \hat{k}^{ad_2})$ or $T =
(\hat{k}^{d_3}, \hat{k}^{d_4})$ with random choices of $a, d_1, d_2, d_3, d_4
\in \Z_p$. It is easy to prove by simple hybrid arguments that if there
exists an adversary that breaks the parallel-SXDH assumption, then it can
break the SXDH assumption. Alternatively, we can tightly prove the reduction
using the random self-reducibility of the Decisional Diffie-Hellman
assumption.

Suppose there exists an adversary $\mc{A}$ that distinguishes between
$\tb{G}_0$ and $\tb{G}_1$ with non-negligible advantage. Simulator $\mc{B}_1$
that solves the parallel-SXDH assumption using $\mc{A}$ is given: a challenge
tuple
    $D = ((p, \G, \hat{\G}, \G_T, e), \lb
    k, \hat{k}, \hat{k}^a, \hat{k}^{d_1}, \hat{k}^{d_2})$ and
    $T = (A_1, A_2)$ where
    $T = T_0 = (A_1^0, A_2^0) = (\hat{k}^{ad_1}, \hat{k}^{ad_2})$ or
    $T = T_1 = (A_1^1, A_2^1) = (\hat{k}^{ad_1 + d_3}, \hat{k}^{ad_2 + d_4})$.
Then $\mc{B}_1$ that interacts with $\mc{A}$ is described as follows:
$\mc{B}_1$ first chooses random exponents $\nu_1, \nu_2, \phi_1, \phi_2,
\phi_3 \in \Z_p$, then it sets $\tau = \phi_1 + \nu_1 \phi_2 + \nu_2 \phi_3$.
It selects random exponents $\alpha, x, y, y_g, y_v, y_w \in \Z_p$ and sets
    $g = k^{y_g}, u = g^x, h = g^y,
    w_1 = k^{y_w \phi_1}, w_2 = k^{y_w \phi_2}, w_3 = k^{y_w \phi_3}, w = k^{y_w},
    \hat{g} = \hat{k}^{y_g}, \hat{u} = \hat{g}^x, \hat{h} = \hat{g}^y$.
It implicitly sets $\nu_3 = a, \pi = \phi_2 + a \phi_3$ and publishes a
public key $PK$ as
    \begin{align*}
    &   g, u, h,~ w_1, w_2, w_3, w,~
        \hat{g}, \hat{g}^{\nu_1}, \hat{g}^{\nu_2}, \hat{g}^{-\tau},~
        \hat{u}, \hat{u}^{\nu_1}, \hat{u}^{\nu_2}, \hat{u}^{-\tau},~ \\
    &   \hat{h}, \hat{h}^{\nu_1}, \hat{h}^{\nu_2}, \hat{h}^{-\tau},~
        \hat{v} = \hat{k}^{y_v}, \hat{v}^{\nu_3} = (\hat{k}^a)^{y_v},
        \hat{v}^{-\pi} = \hat{k}^{-y_v \phi_2} (\hat{k}^a)^{-y_v \phi_3},~
        \Omega = e(g, \hat{g})^{\alpha}.
    \end{align*}
It sets a private key $SK = \alpha$. Additionally, it sets $f = k, \hat{f} =
\hat{k}$ for the semi-functional signature and verification. $\mc{A}$
adaptively requests a signature for a message $M$. To response this sign
query, $\mc{B}_1$ creates a normal signature by calling \textbf{PKS1.Sign}
since it knows the private key. Note that it cannot create a semi-functional
signature since it does not know $k^a$. Finally, $\mc{A}$ outputs a forged
signature $\sigma^* = (W_{1,1}^*, \ldots, W_{2,4}^*)$ on a message $M^*$ from
$\mc{A}$. To verify the forged signature, $\mc{B}_1$ first chooses a random
exponent $t \in \Z_p$ and computes verification components by implicitly
setting $s_1 = d_1,~ s_2 = d_2$ as
    \begin{align*}
    & V_{1,1} = \hat{g}^t,~
      V_{1,2} = (\hat{g}^{\nu_1})^t (\hat{k}^{d_1})^{y_v},~
      V_{1,3} = (\hat{g}^{\nu_2})^t (A_1)^{y_v},~
      V_{1,4} = (\hat{g}^{-\tau})^t (\hat{k}^{d_1})^{-y_v \phi_2} (A_1)^{-y_v \phi_3}, \\
    & V_{2,1} = (\hat{u}^{M^*} \hat{h})^t,~
      V_{2,2} = ((\hat{u}^{\nu_1})^{M^*} \hat{h}^{\nu_1})^t (\hat{k}^{d_2})^{y_v},~
      V_{2,3} = ((\hat{u}^{\nu_2})^{M^*} \hat{h}^{\nu_2})^t (A_2)^{y_v},~ \\
    & V_{2,4} = ((\hat{u}^{-\tau})^{M^*} \hat{h}^{-\tau})^t
                (\hat{k}^{d_2})^{-y_v \phi_2} (A_2)^{-y_v \phi_3}.
    \end{align*}
Next, it verifies that
    $\prod_{i=1}^4 e(W_{1,i}^*, V_{1,i}) \cdot
    \prod_{i=1}^4  e(W_{2,i}^*, V_{2,i})^{-1} \stackrel{?}{=} \Omega^t$.
If this equation holds, then it outputs 0. Otherwise, it outputs 1.

\vs To finish this proof, we show that the distribution of the simulation is
correct. We first show that the distribution using $D, T_0 = (A_1^0, A_2^0) =
(\hat{k}^{ad_1}, \hat{k}^{ad_2})$ is the same as $\tb{G}_0$. The public key
is correctly distributed since the random blinding values $y_g, y_w, y_v$ are
used. The signatures is correctly distributed since it uses the signing
algorithm. The verification components are correctly distributed as
    \begin{align*}
    V_{1,3} &= (\hat{g}^{\nu_2})^t (\hat{v}^{\nu_3})^{s_1}
             = (\hat{g}^{\nu_2})^t (\hat{k}^{y_v a})^{d_1}
             = (\hat{g}^{\nu_2})^t (A_1^0)^{y_v}, \\
    V_{1,4} &= (\hat{g}^{-\tau})^t (\hat{v}^{-\pi})^{s_1}
             = (\hat{g}^{-\tau})^t (\hat{k}^{-y_v (\phi_2 + a \phi_3)})^{d_1}
             = (\hat{g}^{-\tau})^t (\hat{k}^{d_1})^{-y_v \phi_2} (A_1^0)^{-y_v \phi_3},
             \db \\
    V_{2,3} &= ((\hat{u}^{\nu_2})^{M^*} \hat{h}^{\nu_2})^t (\hat{v}^{\nu_3})^{s_2}
             = ((\hat{u}^{\nu_2})^{M^*} \hat{h}^{\nu_2})^t (\hat{k}^{y_v a})^{d_2}
             = ((\hat{u}^{\nu_2})^{M^*} \hat{h}^{\nu_2})^t (A_2^0)^{y_v} \\
    V_{2,4} &= ((\hat{u}^{-\tau})^{M^*} \hat{h}^{-\tau})^t (\hat{v}^{-\pi})^{s_2}
             = ((\hat{u}^{-\tau})^{M^*} \hat{h}^{-\tau})^t
               (\hat{k}^{-y_v (\phi_2 + a \phi_3)})^{d_2} \\
            &= ((\hat{u}^{-\tau})^{M^*} \hat{h}^{-\tau})^t
               (\hat{k}^{d_2})^{-y_v \phi_2} (A_2^0)^{-y_v \phi_3}.
    \end{align*}
We next show that the distribution of the simulation using $D, T_1 = (A_1^1,
A_2^1) = (\hat{k}^{ad_1 + d_3}, \hat{k}^{ad_2 + d_4})$ is the same as
$\tb{G}_1$. We only consider the distribution of the verification components
since $T$ is only used in the verification components. The difference between
$T_0 = (A_1^0, A_2^0)$ and $T_1 = (A_1^1, A_2^1)$ is that $T_1 = (A_1^1,
A_2^1)$ additionally has $(\hat{k}^{d_3}, \hat{k}^{d_4})$. Thus $V_{1,3},
V_{1,4}, V_{2,3}, V_{2,4}$ that have $T = (A_1, A_2)$ in the simulation
additionally have $(\hat{k}^{d_3})^{y_v}, (\hat{k}^{d_3})^{-y_v \phi_3},
(\hat{k}^{d_4})^{y_v}, (\hat{k}^{d_4})^{-y_v \phi_3}$ respectively. If we
implicitly set $s_c = y_v d_3,~ z_c = d_4 / d_3$, then the verification
components for the forged signature are semi-functional since $d_3, d_4$ are
randomly chosen.
%
%We obtain $\Pr [\mc{B}_1(D,T_0) = 0] = \Adv_{\mc{A}}^{G_0}$ and $\Pr
%[\mc{B}_1(D,T_1) = 0] = \Adv_{\mc{A}}^{G_1}$ from the above analysis. Thus,
%we can easily derive the advantage of $\mc{B}_1$ as
%    \begin{eqnarray*}
%    \Adv_{\mc{B}_1}^{SXDH}(\lambda)
%        = \big| \Pr[\mc{B}_1(D, T_0) = 0] - \Pr[\mc{B}_1(D, T_1) = 0] \big|
%        = \big| \Adv_{\mc{A}}^{G_0} - \Adv_{\mc{A}}^{G_1} \big|.
%    \end{eqnarray*}
This completes our proof.
\end{proof}

\begin{lemma} \label{lem:pks1-prime-2}
If the LW2 assumption holds, then no polynomial-time adversary can
distinguish between $\tb{G}_1$ and $\tb{G}_2$ with non-negligible advantage.
That is, for any adversary $\mc{A}$, there exists a PPT algorithm $\mc{B}_2$
such that
    $\big| \Adv_{\mc{A}}^{G_{1,k-1}} - \Adv_{\mc{A}}^{G_{1,k}} \big|
    = \Adv_{\mc{B}_2}^{LW2}(\lambda)$.
\end{lemma}

\begin{proof}
Suppose there exists an adversary $\mc{A}$ that distinguishes between
$\tb{G}_{1,k-1}$ and $\tb{G}_{1,k}$ with non-negligible advantage. A
simulator $\mc{B}_2$ that solves the LW2 assumption using $\mc{A}$ is given:
a challenge tuple
    $D = ((p, \G, \hat{\G}, \G_T, e), \lb
    k, k^a, k^b, k^c, \hat{k}^a, \hat{k}^{a^2}, \hat{k}^{bx}, \hat{k}^{abx},
    \hat{k}^{a^2x})$ and $T$
where $T = T_0 = k^{bc}$ or $T = T_1 = k^{bc+d}$. Then $\mc{B}_2$ that
interacts with $\mc{A}$ is described as follows: $\mc{B}_2$ first selects
random exponents $\nu_1, \nu_2, \nu_3, y_{\tau}, \pi, A, B, \alpha, y_u, y_h,
\lb y_w, y_v \in \Z_p$ and sets
    $g = k^a, u = (k^a)^A k^{y_u}, h = (k^a)^B k^{y_h}, w = k^{y_w},
    \hat{g} = \hat{k}^a, \hat{u} = (\hat{k}^a)^A \hat{k}^{y_u},
    \hat{h} = (\hat{k}^a)^B \hat{k}^{y_h}, \hat{v} = \hat{k}^{y_v}$.
It implicitly sets $\phi_1 = (\nu_1 \nu_3 - \nu_2) b - \nu_1 \pi + (a +
y_{\tau}), \phi_2 = -\nu_3 b + \pi, \phi_3 = b, \tau = a + y_{\tau}$ and
publishes a public key $PK$ as
    \begin{align*}
    &   g, u, h,~
        w_1 = ((k^b)^{\nu_1 \nu_3 - \nu_2} k^{-\nu_1 \pi} (k^a) k^{y_{\tau}})^{y_w},
        w_2 = ((k^b)^{-\nu_3} k^{\pi})^{y_w}, w_3 = (k^b)^{y_w}, w,~ \\
    &   \hat{g}, \hat{g}^{\nu_1}, \hat{g}^{\nu_2},
        \hat{g}^{-\tau} = (\hat{k}^{a^2} (\hat{k}^a)^{y_{\tau}})^{-1}),~
        \hat{u}, \hat{u}^{\nu_1}, \hat{u}^{\nu_2},
        \hat{u}^{-\tau} = ((\hat{k}^{a^2})^A (\hat{k}^a)^{y_u + A y_{\tau}}
            \hat{k}^{y_u y_{\tau}})^{-1},~ \\
    &   \hat{h}, \hat{h}^{\nu_1}, \hat{h}^{\nu_2},
        \hat{h}^{-\tau} = ((\hat{k}^{a^2})^B (\hat{k}^a)^{y_h + B y_{\tau}}
            \hat{k}^{y_h y_{\tau}} )^{-1},~
        \hat{v}, \hat{v}^{\nu_3}, \hat{v}^{-\pi},~
        \Omega = e(k^a, \hat{k}^a)^{\alpha}.
    \end{align*}
Additionally, it sets $f = k, \hat{f} = \hat{k}$ for the semi-functional
signature and verification. $\mc{A}$ adaptively requests a signature for a
message $M$. If this is a $j$-th signature query, then $\mc{B}_2$ handles
this query as follows:
\begin{itemize}
\item {Case $j < k$} : It creates a semi-functional signature by calling
    \textbf{PKS1.SignSF} since it knows the tuple $(f^{\nu_1 \nu_3 -
    \nu_2}, f^{-\nu_3}, f, 1)$ for the semi-functional signature.

\item {Case $j = k$} : It selects random exponents $r', c'_1, c'_2 \in
    \Z_p$ and creates a signature by implicitly setting $r = -c + r',~
    c_1 = c(AM+B)/y_w + c'_1,~ c_2 = c/y_w + c'_2$ as
    \begin{align*}
    &   W_{1,1} = g^{\alpha} (k^c)^{-(y_u M + y_h)} (u^M h)^{r'}
                  (T)^{(\nu_1 \nu_3 - \nu_2)(AM+B)}
                  (k^c)^{(-\nu_1 \pi + y_{\tau})(AM+B)} w_1^{c'_1},~ \\
    &   W_{1,2} = (T)^{-\nu_3 (AM+B)} (k^c)^{\pi (AM+B)} w_2^{c'_1},~
        W_{1,3} = (T)^{(AM+B)} w_3^{c'_1},~
        W_{1,4} = (k^c)^{(AM+B)} w^{c'_1},~\\
    &   W_{2,1} = g^{r'} (T)^{(\nu_1 \nu_3 - \nu_2)}
                  (k^c)^{(-\nu_1 \pi + y_{\tau})} w_1^{c'_2},~
        W_{2,2} = (T)^{-\nu_3} (k^c)^{y_w \pi} w_2^{c'_2},~
        W_{2,3} = T w_3^{c'_2},~
        W_{2,4} = (k^c)^{y_w} w^{c'_2}.
    \end{align*}

\item {Case $j > k$} : It creates a normal signature by calling
    \textbf{PKS1.Sign} since it knows $\alpha$ of the private key. Note
    that $x, y$ are not required.
\end{itemize}

\noindent Finally, $\mc{A}$ outputs a forged signature $\sigma^* =
(W_{1,1}^*, \ldots, W_{2,4}^*)$ on a message $M^*$. To verify the forged
signature, $\mc{B}_2$ first chooses random exponents $t', s_1, s_2 \in \Z_p$
and computes semi-functional verification components by implicitly setting
    $t = bx + t',~ s_c = -a^2 x,~ z_c = AM^* + B$
as
    \begin{align*}
    &   V_{1,1} = \hat{k}^{abx} (\hat{k}^a)^{t'},~
        V_{1,2} = (\hat{k}^{abx})^{\nu_1} (\hat{k}^a)^{\nu_1 t'} \hat{v}^{s_1},~ \\
    &   V_{1,3} = (\hat{k}^{abx})^{\nu_2} (\hat{k}^a)^{\nu_2 t'} \hat{v}^{\nu_3 s_1}
                  (\hat{k}^{a^2 x})^{-1},~
        V_{1,4} = (\hat{k}^{abx})^{-y_{\tau}} (\hat{k}^{a^2})^{-t'}
                  (\hat{k}^a)^{-y_{\tau} t'} \hat{v}^{-\pi s_1}, \db \\
    &   V_{2,1} = (\hat{k}^{abx})^{AM^* + B} (\hat{k}^{bx})^{y_u M^* + y_h}
                  (\hat{u}^{M^*} \hat{h})^{t'},~ \\
    &   V_{2,2} = (\hat{k}^{abx})^{(AM^* + B) \nu_1}
                  (\hat{k}^{bx})^{(y_u M^* + y_h) \nu_1}
                  (\hat{u}^{M^*} \hat{h})^{\nu_1 t'} \hat{v}^{s_2},~ \\
    &   V_{2,3} = (\hat{k}^{abx})^{(AM^* + B) \nu_2}
                  (\hat{k}^{bx})^{(y_u M^* + y_h) \nu_2}
                  (\hat{u}^{M^*} \hat{h})^{\nu_2 t'} \hat{v}^{\nu_3 s_2}
                  (\hat{k}^{a^2 x})^{-(AM^* + B)},~ \\
    &   V_{2,4} = (\hat{k}^{abx})^{-(AM^* + B) y_{\tau} -(y_u M^* + y_h)}
                  (\hat{k}^{bx})^{-(y_u M^* + y_h) y_{\tau}}
                  (\hat{k}^{a^2})^{-(AM^* + B) t'} (\hat{k}^a)^{-(y_u M^* + y_h) t'}
                  (\hat{u}^{M^*} \hat{h})^{-y_{\tau} t'} \hat{v}^{-\pi s_2}.
    \end{align*}
Next, it verifies that
    $\prod_{i=1}^4 e(W_{1,i}^*, V_{1,i}) \cdot
    \prod_{i=1}^4 e(W_{2,i}^*, V_{2,i})^{-1} \stackrel{?}{=}
    e(k^a, \hat{k}^{abx})^{\alpha} \cdot e(k^a, \hat{k}^a)^{\alpha t'}$.
If this equation holds, then it outputs 0. Otherwise, it outputs 1.

\vs To finish the proof, we should show that the distribution of the
simulation is correct. We first show that the distribution of the simulation
using $D, T_0 = k^{bc}$ is the same as $\tb{G}_{1,k-1}$. The public key is
correctly distributed since the random blinding values $y_u, y_h, y_w, y_v$
are used. The $k$-th signature is correctly distributed as
    \begin{align*}
    W_{1,1} &= g^{\alpha} (u^M h)^r w_1^{c_1}
             = g^{\alpha} (k^{(aA + y_u) M} k^{aB + y_h})^{-c + r'}
               (k^{y_w ((\nu_1 \nu_3 - \nu_2) b - \nu_1 \pi + (a + y_{\tau}))})^{
                    c (AM+B)/y_w + c'_1} \\
            &= g^{\alpha} (k^c)^{-(y_u M + y_h)} (u^M h)^{r'}
               (T)^{(\nu_1 \nu_3 - \nu_2) (AM+B)}
               (k^c)^{(-\nu_1 \pi + y_{\tau}) (AM+B)} w_1^{c'_1},~
            \displaybreak[0] \\
    W_{1,2} &= w_2^{c_1}
             = (k^{y_w (-\nu_3 b + \pi)})^{c (AM+B)/y_w + c'_1}
             = (T)^{-\nu_3 (AM+B)} (k^c)^{\pi (AM+B)} w_2^{c'_1},~ \\
    W_{1,3} &= w_3^{c_1}
             = (k^{y_w b})^{c (AM+B)/y_w + c'_1}
             = (T)^{(AM+B)} w_3^{c'_1},~ \\
    W_{1,4} &= w^{c_1}
             = (k^{y_w})^{c (AM+B)/y_w + c'_1}
             = (k^c)^{(AM+B)} w^{c'_1}.
    \end{align*}
The semi-functional verification components are correctly distributed as
    \begin{align*}
    V_{2,1} &= (\hat{u}^{M^*} \hat{h})^t
             = (\hat{k}^{(aA + y_u) M^*} \hat{k}^{aB + y_h})^{bx + t'}
             = (\hat{k}^{abx})^{AM^* + B} (\hat{k}^{bx})^{y_u M^* + y_h}
               (\hat{u}^{M^*} \hat{h})^{t'},~ \\
    V_{2,2} &= ((\hat{u}^{\nu_1})^{M^*} \hat{h}^{\nu_1})^t \hat{v}^{s_2}
             = (\hat{k}^{(aA + y_u) \nu_1 M^*} \hat{k}^{(aB + y_h) \nu_1})^{bx + t'}
               \hat{v}^{s_2} \\
            &= (\hat{k}^{abx})^{(AM^* + B) \nu_1} (\hat{k}^{bx})^{(y_u M^* + y_h) \nu_1}
               (\hat{u}^{M^*} \hat{h})^{\nu_1 t'} \hat{v}^{s_2},~
               \displaybreak[0] \\
    V_{2,3} &= ((\hat{u}^{\nu_2})^{M^*} \hat{h}^{\nu_2})^t (\hat{v}^{\nu_3})^{s_2}
               \hat{f}^{s_c z_c}
             = (\hat{k}^{(aA + y_u) \nu_2 M^*} \hat{k}^{(aB + y_h) \nu_2})^{bx + t'}
               (\hat{v}^{\nu_3})^{s_2} \hat{k}^{-a^2x (AM^*+B)} \\
            &= (\hat{k}^{abx})^{(AM^* + B) \nu_2} (\hat{k}^{bx})^{(y_u M^* + y_h) \nu_2}
               (\hat{u}^{M^*} \hat{h})^{\nu_2 t'} \hat{v}^{\nu_3 s_2}
               (\hat{k}^{a^2 x})^{-(AM^* + B)},~
               \displaybreak[0] \\
    V_{2,4} &= ((\hat{u}^{-\tau})^{M^*} \hat{h}^{-\tau})^t (\hat{v}^{-\pi})^{s_2}
               (\hat{f}^{-\phi_3})^{s_c z_c} \\
            &= (\hat{k}^{-(aA + y_u) (a + y_{\tau}) M^*}
               \hat{k}^{-(aB + y_h) (a + y_{\tau})})^{bx + t'}
               (\hat{v}^{-\pi})^{s_2} \hat{k}^{-b (-a^2 x) (AM^*+B)} \\
            &= (\hat{k}^{abx})^{-(AM^* + B) y_{\tau} -(y_u M^* + y_h)}
               (\hat{k}^{bx})^{-(y_u M^* + y_h) y_{\tau}}
               (\hat{k}^{a^2})^{-(AM^* + B) t'} (\hat{k}^a)^{-(y_u M^* + y_h) t'}
               (\hat{u}^{M^*} \hat{h})^{-y_{\tau} t'} \hat{v}^{-\pi s_2}.
    \end{align*}
The simulator can create the semi-functional verification components with
only fixed $z_c = AM^* + B$ since $s_c, s_c$ enable the cancellation of
$\hat{k}^{a^2 b x}$. Even though the simulator uses the fixed $z_c$, the
distribution of $z_c$ is correct since $A,B$ are information theoretically
hidden to $\mc{A}$.
We next show that the distribution of the simulation using $D, T_1 =
k^{bc+d}$ is the same as $\tb{G}_{1,k}$. We only consider the distribution of
the $k$-th signature since $T$ is only used in the $k$-th signature. The only
difference between $T_0$ and $T_1$ is that $T_1$ additionally has $k^d$. The
signature components $W_{1,1}, W_{1,2}, W_{1,3}$, $W_{2,1}, W_{2,2}, W_{2,3}$
that have $T$ in the simulation additionally have $(k^d)^{(\nu_1 \nu_3 -
\nu_2) (AM+B)}$, $(k^d)^{-\nu_3 (AM+B)}$, $(k^d)^{(AM+B)}$, $(k^d)^{(\nu_1
\nu_3 - \nu_2)}$, $(k^d)^{-\nu_3}, k^d$ respectively. If we implicitly set
$s_k = d, z_k = AM+B$, then the distribution of the $k$-th signature is the
same as $\tb{G}_{1,k}$ except that the $k$-th signature is nominally
semi-functional.

Finally, we show that the adversary cannot distinguish the nominally
semi-functional signature from the semi-functional signature. The main idea
of this is that the adversary cannot request a signature for the forgery
message $M^*$ in the security model. Suppose there exists an unbounded
adversary, then the adversary can gather $z_k = AM + B$ from the $k$-th
signature and $z_c = AM^* + B$ from the forged signature. It is easy to show
that $z_k$ and $z_c$ look random to the unbounded adversary since $f(M) = AM
+ B$ is a pair-wise independent function and $A, B$ are information
theoretically hidden to the adversary.
%
%We obtain $\Pr [\mc{B}_2(D,T_0) = 0] = \Adv_{\mc{A}}^{G_{1,k-1}}$ and $\Pr
%[\mc{B}_2(D,T_1) = 0] = \Adv_{\mc{A}}^{G_{1,k}}$ from the above analysis.
%Thus, we can easily derive the advantage of $\mc{B}_2$ as
%    \begin{eqnarray*}
%    \Adv_{\mc{B}_2}^{LW2}(\lambda)
%     =  \big| \Pr[\mc{B}_2(D, T_0) = 0] - \Pr[\mc{B}_2(D, T_1) = 0] \big|
%     =  \big| \Adv_{\mc{A}}^{G_{1,k-1}} - \Adv_{\mc{A}}^{G_{1,k}} \big|.
%    \end{eqnarray*}
This completes our proof.
\end{proof}

\begin{lemma} \label{lem:pks1-prime-3}
If the DBDH assumption holds, then no polynomial-time adversary can
distinguish between $\tb{G}_2$ and $\tb{G}_3$ with non-negligible advantage.
That is, for any adversary $\mc{A}$, there exists a PPT algorithm $\mc{B}_3$
such that
    $\big| \Adv_{\mc{A}}^{G_2} - \Adv_{\mc{A}}^{G_3} \big| =
    \Adv_{\mc{B}_3}^{DBDH}(\lambda)$.
\end{lemma}

\begin{proof}
Suppose there exists an adversary $\mc{A}$ that distinguish $\tb{G}_2$ from
$\tb{G}_3$ with non-negligible advantage. A simulator $\mc{B}_3$ that solves
the DBDH assumption using $\mc{A}$ is given: a challenge tuple
    $D = ((p, \G, \hat{\G}, \G_T, e), \lb
    k, k^a, k^b, k^c, \hat{k}, \hat{k}^a, \hat{k}^b, \hat{k}^c)$
    and $T$ where $T = T_0 = e(k, \hat{k})^{abc}$ or $T = T_1 = e(k, \hat{k})^d$.
Then $\mc{B}_3$ that interacts with $\mc{A}$ is described as follows:
$\mc{B}_3$ first chooses random exponents $\nu_1, \nu_3$, $\phi_1, \phi_2,
\phi_3 \in \Z_p$ and sets $\pi = \phi_2 + \nu_3 \phi_3$. It selects random
exponents $y_g, x, y, y_w, y_v \in \Z_p$ and sets
    $g = k^{y_g}, u = g^{x}, h = g^{y},
    w_1 = k^{y_w \phi_1}, w_2 = k^{y_w \phi_2}, w_3 = k^{y_w \phi_3}, w = k^{y_w},
    \hat{g} = \hat{k}^{y_g}, \hat{u} = \hat{g}^{x}, \hat{h} = \hat{g}^{y},
    \hat{v} = \hat{k}^{y_v}$.
It implicitly sets $\nu_2 = a, \tau = \phi_1 + \nu_1 \phi_2 + a \phi_3,
\alpha = ab$ and publishes a public key $PK$ as
    \begin{align*}
    &   g, u, h,~ w_1, w_2, w_3, w,~
        \hat{g}, \hat{g}^{\nu_1}, \hat{g}^{\nu_2} = (\hat{k}^a)^{y_g},
        \hat{g}^{-\tau} = \hat{k}^{-y_g (\phi_1 + \nu_1 \phi_2)}
            (\hat{k}^a)^{-y_g \phi_3},~ \\
    &   \hat{u}, \hat{u}^{\nu_1}, \hat{u}^{\nu_2} = (\hat{g}^{\nu_2})^x,
        \hat{u}^{-\tau} = (\hat{g}^{-\tau})^x,~
        \hat{h}, \hat{h}^{\nu_1}, \hat{h}^{\nu_2} = (\hat{g}^{\nu_2})^y,
        \hat{h}^{-\tau} = (\hat{g}^{-\tau})^y,~ \\
    &   \hat{v}, \hat{v}^{\nu_3}, \hat{v}^{-\pi},~
        \Omega = e(k^a, \hat{k}^b)^{y_g^2}.
    \end{align*}
Additionally, it sets $f = k, \hat{f} = \hat{k}$ for the semi-functional
signature and semi-functional verification. $\mc{A}$ adaptively requests a
signature for a message $M$. To respond to this query, $\mc{B}_3$ selects
random exponents $r, c_1, c_2, s_k, z'_k \in \Z_p$ and creates a
semi-functional signature by implicitly setting $z_k = b y_g / s_k + z'_k$ as
    \begin{align*}
    &   W_{1,1} = (u^M h)^r w_1^{c_1} (k^b)^{\nu_1 \nu_3 y_g}
                  k^{\nu_1 \nu_3 s_k z'_k} (k^a)^{-s_k z'_k},~ \\
    &   W_{1,2} = w_2^{c_1} (k^b)^{-\nu_3 y_g} k^{-\nu_3 s_k z'_k},~
        W_{1,3} = w_3^{c_1} (k^b)^{y_g} k^{s_k z'_k},~
        W_{1,4} = w^{c_1}, \\
    &   W_{2,1} = g^r w_1^{c_2} k^{\nu_1 \nu_3 s_k} (k^a)^{-s_k},~
        W_{2,2} = w_2^{c_2} k^{-\nu_3 s_k},~
        W_{2,3} = w_3^{c_2} k^{s_k},~
        W_{2,4} = w^{c_2}.
    \end{align*}
The simulator can only create a semi-functional signature since $s_k, z_k$
enables the cancellation of $k^{ab}$. Finally, $\mc{A}$ outputs a forged
signature $\sigma^* = (W_{1,1}^*, \ldots, W_{2,4}^*)$ on a message $M^*$. To
verify the forged signature, $\mc{B}_3$ first chooses random exponents $s_1,
s_2, s'_c, z'_c \in \Z_p$ and computes semi-functional verification
components by implicitly setting
    $t = c,~ s_c = -ac y_g + s'_c,~
    z_c = -ac y_g (xM^*+y)/s_c + z'_c/s_c$
as
    \begin{align*}
    &   V_{1,1} = (\hat{k}^c)^{y_g},~
        V_{1,2} = (\hat{k}^c)^{y_g \nu_1} \hat{v}^{s_1},~
        V_{1,3} = \hat{v}^{\nu_3 s_1} \hat{k}^{s'_c},~
        V_{1,4} = (\hat{k}^c)^{-y_g (\phi_1 + \nu_1 \phi_2)} \hat{v}^{-\pi s_1}
                  \hat{k}^{-\phi_3 s'_c}, \\
    &   V_{2,1} = (\hat{k}^c)^{y_g (xM^*+y)},~
        V_{2,2} = (\hat{k}^c)^{y_g (xM^*+y) \nu_1} \hat{v}^{s_2},~
        V_{2,3} = \hat{v}^{\nu_3 s_2} \hat{k}^{z'_c},~ \\
    &   V_{2,4} = (\hat{k}^c)^{-y_g (xM^*+y) (\phi_1 + \nu_1 \phi_2)}
                  \hat{v}^{-\pi s_2} \hat{k}^{-\phi_3 z'_c}.
    \end{align*}
Next, it verifies that
    $\prod_{i=1}^4 e(W_{1,i}^*, V_{1,i}) \cdot \prod_{i=1}^4 e(W_{2,i}^*, V_{2,i})^{-1}
    \stackrel{?}{=} (T)^{y_g^2}.$
If this equation holds, then it outputs $0$. Otherwise, it outputs $1$.

\vs To finish the proof, we first show that the distribution of the
simulation using $D, T = e(k,\hat{k})^{abc}$ is the same as $\tb{G}_2$. The
public key is correctly distributed since the random blinding values $y_g,
y_w, y_v$ are used. The semi-functional signature is correctly distributed as
    \begin{align*}
    W_{1,1} &= g^{\alpha} (u^M h)^r w_1^{c_1} (f^{\nu_1 \nu_3 - \nu_2})^{s_k z_k}
             = k^{y_g ab} (u^M h)^r w_1^{c_1}
               (k^{\nu_1 \nu_3 - a})^{s_k (b y_g / s_k + z'_k)} \\
            &= (u^M h)^r w_1^{c_1} (k^b)^{\nu_1 \nu_3 y_g} k^{\nu_1 \nu_3 s_k z'_k}
               (k^a)^{-s_k z'_k}.
    \end{align*}
The semi-functional verification components are correctly distributed as
    \begin{align*}
    V_{1,3} &= (\hat{g}^{\nu_2})^t (\hat{v}^{\nu_3})^{s_1} \hat{f}^{s_c}
             = (\hat{k}^{y_g a})^c \hat{v}^{\nu_3 s_1} \hat{k}^{-ac y_g + s'_c}
             = \hat{v}^{\nu_3 s_1} \hat{k}^{s'_c},~ \\
    V_{1,4} &= (\hat{g}^{-\tau})^t (\hat{v}^{-\pi})^{s_1} (\hat{f}^{-\phi_3})^{s_c}
             = (\hat{k}^{-y_g (\phi_1 + \nu_1 \phi_2 + a \phi_3)})^c \hat{v}^{-\pi s_1}
               \hat{k}^{-\phi_3 (-ac y_g + s'_c)}
             = (\hat{k}^c)^{-y_g (\phi_1 + \nu_1 \phi_2)} \hat{v}^{-\pi s_1}
               \hat{k}^{-\phi_3 s'_c}, \db \\
    V_{2,3} &= (\hat{u}^{\nu_2 M^*} \hat{h}^{\nu_2})^t (\hat{v}^{\nu_3})^{s_2}
               \hat{f}^{s_c z_c}
             = (\hat{k}^{y_g a (xM^*+y)})^c (\hat{v}^{\nu_3})^{s_2}
               \hat{k}^{-ac y_g (xM^*+y) + z'_c}
             = \hat{v}^{\nu_3 s_2} \hat{k}^{z'_c},~ \\
    V_{2,4} &= (\hat{u}^{-\tau M^*} \hat{h}^{-\tau})^t (\hat{v}^{-\pi})^{s_2}
               (\hat{f}^{-\phi_3})^{s_c z_c}
             = (\hat{k}^{-y_g (\phi_1 + \nu_1 \phi_2 + a \phi_3) (xM^* + y)})^c
               (\hat{v}^{-\pi})^{s_2} (\hat{k}^{-\phi_3})^{-ac y_g (xM^*+y) + z'_c} \\
            &= (\hat{k}^c)^{-y_g (xM^*+y) (\phi_1 + \nu_1 \phi_2)} \hat{v}^{-\pi s_2}
               \hat{k}^{-\phi_3 z'_c}, \\
    \Omega^t &= e(g, \hat{g})^{\alpha t}
             = e(k, \hat{k})^{y_g^2 ab c} = (T_0)^{y_g^2} .
    \end{align*}
We next show that the distribution of the simulation using $D, T_1 = e(k,
\hat{k})^d$ is almost the same as $\tb{G}_3$. It is obvious that the
signature verification for the forged signature always fails if $T_1 = e(k,
\hat{k})^d$ is used except with $1/p$ probability since $d$ is a random value
in $\Z_p$.
%
%We obtain $\Pr [\mc{B}_3(D,T_0) = 0] = \Adv_{\mc{A}}^{G_2}$ and $\Pr
%[\mc{B}_3(D,T_1) = 0] = \Adv_{\mc{A}}^{G_3}$ from the above analysis. Thus,
%we can easily derive the advantage of $\mc{B}_3$ as
%    \begin{eqnarray*}
%    \Adv_{\mc{B}_3}^{DBDH}(\lambda)
%     =  \big| \Pr[\mc{B}_3(D, T_0) = 0] - \Pr[\mc{B}_3(D, T_1) = 0] \big|
%     =  \big| \Adv_{\mc{A}}^{G_2} - \Adv_{\mc{A}}^{G_3} \big|.
%    \end{eqnarray*}
This completes our proof.
\end{proof}

\subsubsection{Analysis of PKS2}

\begin{theorem} \label{thm:pks2-prime}
The above \tb{PKS2} scheme is existentially unforgeable under a chosen
message attack if the LW1, LW2, DBDH assumptions hold. That is, for any PPT
adversary $\mc{A}$, there exist PPT algorithms $\mc{B}_1, \mc{B}_2, \mc{B}_3$
such that
    $\Adv_{\mc{A}}^{PKS}(\lambda)
    \leq \Adv_{\mc{B}_1}^{LW1}(\lambda) + q \Adv_{\mc{B}_2}^{LW2}(\lambda) +
        \Adv_{\mc{B}_3}^{DBDH}(\lambda)$
where $q$ is the maximum number of signature queries of $\mc{A}$.
\end{theorem}

\begin{proof}
Before proving the security, we first define two additional algorithms for
semi-functional types. For the semi-functionality, we set $f = g^{y_f},
\hat{f} = \hat{g}^{y_f}$ where $y_f$ is a random exponent in $\Z_p$.

\begin{description}
\item [\tb{PKS2.SignSF}.] The semi-functional signing algorithm first
    creates a normal signature using the private key. Let $(W'_{1,1},
    \ldots, W'_{2,3})$ be the normal signature of a message $M$ with random
    exponents $r, c_1, c_2 \in \Z_p$. It selects random exponents $s_k, z_k
    \in \Z_p$ and outputs a semi-functional signature as
    \begin{align*}
    \sigma = \Big(~
    &   W_{1,1} = W'_{1,1} \cdot (f^{-{\nu}})^{s_k z_k},~
        W_{1,2} = W'_{1,2} \cdot f^{s_k z_k},~
        W_{1,3} = W'_{1,3},~\\
    &   W_{2,1} = W'_{2,1} \cdot (f^{-{\nu}})^{s_k},~
        W_{2,2} = W'_{2,2} \cdot f^{s_k},~
        W_{2,3} = W'_{2,3}
    ~\Big).
    \end{align*}

\item [\tb{PKS2.VerifySF}.] The semi-functional verification algorithm
    first creates normal verification components using the public key. Let
    $(V'_{1,1}, \ldots, V'_{2,3})$ be the normal verification components
    with a random exponent $t \in \Z_p$. It chooses random exponents $s_c,
    z_c \in \Z_p$ and computes semi-functional verification components as
    \begin{align*}
    &   V_{1,1} = V'_{1,1},~
        V_{1,2} = V'_{1,2} \cdot \hat{f}^{s_c},~
        V_{1,3} = V'_{1,3} \cdot (\hat{f}^{-\phi_2})^{s_c}, \\
    &   V_{2,1} = V'_{2,1},~
        V_{2,2} = V'_{2,2} \cdot \hat{f}^{s_c z_c},~
        V_{2,3} = V'_{2,3} \cdot (\hat{f}^{-\phi_2})^{s_c z_c}.
    \end{align*}
    Next, it verifies that $\prod_{i=1}^3 e(W_{1,i}, V_{1,i}) \cdot
    \prod_{i=1}^3 e(W_{2,i}, V_{2,i})^{-1} \stackrel{?}{=} \Omega^t$. If
    this equation holds, then it outputs 1. Otherwise, it outputs 0.
\end{description}
If the semi-functional verification algorithm is used to verify a
semi-functional signature, then an additional random element $e(f,
\hat{f})^{s_k s_c (z_k - z_c)}$ is left in the left part of the above
verification equation. If $z_k = z_c$, then the semi-functional verification
algorithm succeeds. In this case, we say that the signature is
\textit{nominally} semi-functional.

The security proof uses a sequence of games $\tb{G}_0, \tb{G}_1, \tb{G}_2$,
and $\tb{G}_3$. The definition of these games is the same as that of Theorem
\ref{thm:pks1-prime}.
From the following three lemmas, we prove that it is hard for $\mc{A}$ to
distinguish $\tb{G}_{i-1}$ from $\tb{G}_{i}$ under the given assumptions.
Therefore, we have that
    \begin{align*}
    \Adv_{\mc{A}}^{PKS}(\lambda)
    & = \Adv_{\mc{A}}^{G_0} +
        \sum_{i=1}^2 \big( \Adv_{\mc{A}}^{G_i} - \Adv_{\mc{A}}^{G_i} \big)
        - \Adv_{\mc{A}}^{G_3}
    \leq \sum_{i=1}^3 \big| \Adv_{\mc{A}}^{G_{i-1}} - \Adv_{\mc{A}}^{G_i} \big| \\
    & = \Adv_{\mc{B}_1}^{LW1}(\lambda) +
        \sum_{k=1}^q \Adv_{\mc{B}_2}^{LW2}(\lambda) +
        \Adv_{\mc{B}_3}^{DBDH}(\lambda).
    \end{align*}
This completes our proof.
\end{proof}

\begin{lemma} \label{lem:pks2-prime-1}
If the LW1 assumption holds, then no polynomial-time adversary can
distinguish between $\tb{G}_0$ and $\tb{G}_1$ with non-negligible advantage.
That is, for any adversary $\mc{A}$, there exists a PPT algorithm $\mc{B}_1$
such that
    $\big| \Adv_{\mc{A}}^{G_0} - \Adv_{\mc{A}}^{G_1} \big|
    = \Adv_{\mc{B}_1}^{LW1}(\lambda)$.
\end{lemma}

\begin{proof}
%The proof of this lemma is almost similar to the proof of Lemma 1 in
%\cite{LewkoW10} except that the public key is generated differently and the
%proof is employed in the PKS setting.
%
Suppose there exists an adversary $\mc{A}$ that distinguishes between
$\tb{G}_0$ and $\tb{G}_1$ with non-negligible advantage. A simulator
$\mc{B}_1$ that solves the LW1 assumption using $\mc{A}$ is given: a
challenge tuple
    $D = ((p, \G, \hat{\G}, \G_T, e), \lb
    k, k^b, \hat{k}, \hat{k}^a, \hat{k}^b,
    \hat{k}^{ab^2}, \hat{k}^{b^2}, \hat{k}^{b^3}, \hat{k}^c, \hat{k}^{ac},
    \hat{k}^{bc}, \hat{k}^{b^2 c}, \hat{k}^{b^3 c})$ and $T$
    where $T = T_0 = \hat{k}^{ab^2c}$ or $T = T_1 = \hat{k}^{ab^2c + d}$.
Then $\mc{B}_1$ that interacts with $\mc{A}$ is described as follows:
$\mc{B}_1$ first chooses random exponents $\phi_2, A, B, \alpha \in \Z_p$,
random values $y_g, y_u, y_h, y_w \in \Z_p$. It computes $w_1 = w^{\phi_1} =
(k^b)^{y_w}, w_2 = w^{\phi_2} = k^{y_w \phi_2}, w = k^{y_w}$ by implicitly
setting $\phi_1 = b$. It implicitly sets $c_g = -b/y_w + c'_g, c_u = -b A/y_w
+ c'_u, c_h = -b B/y_w + c'_h, \nu = a, \tau = b + a \phi_2$ and publishes a
public key $PK$ by selecting random values $c'_g, c'_u, c'_h \in \Z_p$ as
    \begin{align*}
    &   g w_1^{c_g} = k^{y_g} w_1^{c'_g},~
        w_2^{c_g} = (k^b)^{-b_2} w_2^{c'_g},~
        w^{c_g} = (k^b)^{-1} w^{c'_g},~ \\
    &   u w_1^{c_u} = k^{y_u} w_1^{c'_u},~
        w_2^{c_u} = (k^b)^{-b_2 A} w_2^{c'_u},~
        w^{c_u} = (k^b)^{-A} w^{c'_u},~ \\
    &   h w_1^{c_h} = k^{y_h} w_1^{c'_h},~
        w_2^{c_h} = (k^b)^{-b_2 B} w_2^{c'_h},~
        w^{c_h} = (k^b)^{-B} w^{c'_h},~
        w_1,~ w_2,~ w,~ \db \\
    &   \hat{g} = \hat{k}^{b^2} \hat{k}^{y_g},~
        \hat{g}^{\nu} = \hat{k}^{ab^2} (\hat{k}^a)^{y_g},~
        \hat{g}^{-\tau} = (\hat{k}^{b^3} (\hat{k}^b)^{y_g}
            (\hat{k}^{ab^2})^{b_2} (\hat{k}^a)^{y_g b_2} )^{-1},~ \db \\
    &   \hat{u} = (\hat{k}^{b^2})^A \hat{k}^{y_u},~
        \hat{u}^{\nu} = (\hat{k}^{ab^2})^A (\hat{k}^a)^{y_u},~
        \hat{u}^{-\tau} = ((\hat{k}^{b^3})^A (\hat{k}^b)^{y_u}
            (\hat{k}^{ab^2})^{A b_2} (\hat{k}^a)^{y_u b_2} )^{-1},~ \\
    &   \hat{h} = (\hat{k}^{b^2})^B \hat{k}^{y_h},~
        \hat{h}^{\nu} = (\hat{k}^{ab^2})^B (\hat{k}^a)^{y_h},~
        \hat{h}^{-\tau} = ((\hat{k}^{b^3})^B (\hat{k}^b)^{y_h}
            (\hat{k}^{ab^2})^{B b_2} (\hat{k}^a)^{y_h b_2} )^{-1},~ \\
    &   \Omega = (e(k^{b^3}, \hat{k}^b) \cdot e(k^{b^2}, \hat{k})^{2 y_g}
            \cdot e(k, \hat{k})^{y_g^2})^{\alpha}.
    \end{align*}
It implicitly sets $g = k^{b^2} k^{y_g}, u = (k^{b^2})^A k^{y_u}, h =
(k^{b^2})^B k^{y_h}$, but it cannot create these elements since $k^{b^2}$ is
not given. Additionally, it sets $f = k, \hat{f} = \hat{k}$ for the
semi-functional signature and verification.
$\mc{A}$ adaptively requests a signature for a message $M$. To response this
sign query, $\mc{B}_1$ first selects random exponents $r, c'_1, c'_2 \in
\Z_p$. It implicitly sets
    $c_1 = -b (\alpha + (A M + B) r) / y_w + c'_1,
    c_2 = -b r_1 / y_w + c'_2$
and creates a normal signature as
    \begin{align*}
    &   W_{1,1} = k^{y_g \alpha + (y_u M + y_h)r} (w_1)^{c'_1},~
        W_{1,2} = (W_{1,3})^{\phi_2},~
        W_{1,3} = (k^b)^{-(\alpha + (A M + B)r)} w^{c'_1},~ \\
    &   W_{2,1} = k^{y_g r} (w_1)^{c'_2},~
        W_{2,2} = (W_{2,3})^{\phi_2},~
        W_{2,3} = (k^b)^{-r} w^{c'_2}.
    \end{align*}
Finally, $\mc{A}$ outputs a forged signature $\sigma^* = (W_{1,1}^*, \ldots,
W_{2,3}^*)$ on a message $M^*$ from $\mc{A}$. To verify the forged signature,
$\mc{B}_1$ first chooses a random exponent $t \in \Z_p$ and computes
verification components by implicitly setting $t = c$ as
    \begin{align*}
    &   V_{1,1} = \hat{k}^{b^2 c} (\hat{k}^c)^{y_g},~
        V_{1,2} = T (\hat{k}^{ac})^{y_g},~
        V_{1,3} = ((\hat{k}^{b^3 c}) (\hat{k}^{bc})^{y_g} (T)^{\phi_2}
                  (\hat{k}^{ac})^{y_g \phi_2})^{-1}, \\
    &   V_{2,1} = (\hat{k}^{b^2 c})^{A M^* + B} (\hat{k}^c)^{y_u M^* + y_h},~
        V_{2,2} = (T)^{A M^* + B} (\hat{k}^{ac})^{y_u M^* + y_h},~ \\
    &   V_{2,3} = \big(
                  (\hat{k}^{b^3 c})^{A M^* + B} (\hat{k}^{bc})^{y_u M^* + y_h}
                  (T)^{\phi_2 (A M^* + B)} (\hat{k}^{ac})^{\phi_2 (y_u M^* + y_h)}
                  \big)^{-1}.
    \end{align*}
Next, it verifies that
    $\prod_{i=1}^3 e(W_{1,i}^*, V_{1,i}) \cdot
    \prod_{i=1}^3  e(W_{2,i}^*, V_{2,i})^{-1} \stackrel{?}{=} \Omega^t$.
If this equation holds, then it outputs 0. Otherwise, it outputs 1.

\vs To finish this proof, we show that the distribution of the simulation is
correct. We first show that the distribution using $D, T_0 = \hat{k}^{ab^2c}$
is the same as $\tb{G}_0$. The public key is correctly distributed as
    \begin{align*}
    &   g w_1^{c_g} = (k^{b^2} k^{y_g}) (k^{b y_w})^{-b/y_w + c'_g}
                    = k^{y_g} w_1^{c'_g}.
    \end{align*}
The simulator cannot create $g, u, h$ since $k^{b^2}$ is not given in the
assumption, but it can create $g w_1^{c_g}, u w_1^{c_u}, h w_1^{c_h}$ since
$c_g, c_u, c_h$ can be used to cancel out $k^{b^2}$. The signature is
correctly distributed as
    \begin{align*}
    W_{1,1} &= g^{\alpha} (u^M h)^{r} w_1^{c_1}
             = (k^{b^2 + y_g})^{\alpha} (k^{(b^2 A + y_u)M} k^{b^2 B + y_h})^{r}
               (k^{b y_w})^{-b(\alpha + (A M + B)r) / y_w + c'_1} \\
            &= k^{y_g \alpha + (y_u M + y_h)r} w_1^{c'_1},~ \\
    W_{2,1} &= g^{r} (w^{b_1})^{c_2}
             = (k^{b^2 + y_g})^{r} (k^{b y_w})^{-b r / y_w + c'_2}
             = k^{y_g r} (w^{b_1})^{c'_2}.
    \end{align*}
It can create a normal signature since $c_1, c_2$ enable the cancellation of
$k^{b^2}$, but it cannot create a semi-functional signature since $k^a$ is
not given. The verification components are correctly distributed as
    \begin{align*}
    V_{1,1} &= \hat{g}^t = (\hat{k}^{b^2 + y_g})^c
             = \hat{k}^{b^2 c} (\hat{k}^c)^{y_g},~
    V_{1,2}  = (\hat{g}^{\nu})^t = \hat{k}^{(b^2 + y_g) ac}
             = T_0 (\hat{k}^{ac})^{y_g},~ \\
    V_{1,3} &= (\hat{g}^{-\tau})^t = (\hat{k}^{(b^2 + y_g) (b + a \phi_2) c})^{-1}
             = ((\hat{k}^{b^3 c}) (\hat{k}^{bc})^{y_g} (T_0)^{\phi_2}
               (\hat{k}^{ac})^{y_g \phi_2})^{-1},
             \db \\
    V_{2,1} &= (u^{M^*} h)^t
             = (k^{(b^2 A + y_u) M^*} k^{b^2 B + y_h})^c
             = (k^{b^2 c})^{A M^* + B} (k^c)^{y_u M^* + y_h},~ \\
    V_{2,2} &= ((u^{\nu})^{M^*} h^{\nu})^t
             = (k^{(b^2 A + y_u)a M^*} k^{(b^2 B + y_h)a})^c
             = (T_0)^{A M^* + B} (k^{ac})^{y_u M^* + y_h},~ \\
    V_{2,3} &= ((u^{-\tau})^{M^*} h^{-\tau})^t
             = ((k^{(b^2 A + y_u)(b + a \phi_2) M^*}
               k^{(b^2 B + y_h) (b + a \phi_2)})^c)^{-1} \\
            &= ((k^{b^3 c})^{A M^* + B} (k^{bc})^{y_u M^* + y_h}
               (T_0)^{\phi_2 (A M^* + B)} (k^{ac})^{\phi_2 (y_u M^* + y_h)})^{-1}.
    \end{align*}
We next show that the distribution of the simulation using $D, T_1 =
\hat{k}^{ab^2c + d}$ is the same as $\tb{G}_1$. We only consider the
distribution of the verification components since $T$ is only used in the
verification components. The difference between $T_0$ and $T_1$ is that $T_1$
additionally has $\hat{k}^{d}$. Thus $V_{1,2}, V_{1,3}, V_{2,2}, V_{2,3}$
that have $T$ in the simulation additionally have $\hat{k}^{d},
(\hat{k}^{d})^{\phi_2}, (\hat{k}^{d})^{A M^* + B}, (\hat{k}^{d})^{\phi_2 (A
M^* + B)}$ respectively. If we implicitly set $s_c = d, z_c = A M^* + B$,
then the verification components of the forged signature are semi-functional
since $A$ and $B$ are information-theoretically hidden to the adversary.
%
%We obtain $\Pr [\mc{B}_1(D,T_0) = 0] = \Adv_{\mc{A}}^{G_0}$ and $\Pr
%[\mc{B}_1(D,T_1) = 0] = \Adv_{\mc{A}}^{G_1}$ from the above analysis. Thus,
%we can easily derive the advantage of $\mc{B}_1$ as
%    \begin{eqnarray*}
%    \Adv_{\mc{B}_1}^{LW1}(\lambda)
%        =  \big| \Pr[\mc{B}_1(D, T_0) = 0] - \Pr[\mc{B}_1(D, T_1) = 0] \big|
%        =  \big| \Adv_{\mc{A}}^{G_0} - \Adv_{\mc{A}}^{G_1} \big|.
%    \end{eqnarray*}
This completes our proof.
\end{proof}

\begin{lemma} \label{lem:pks2-prime-2}
If the LW2 assumption holds, then no polynomial-time adversary can
distinguish between $\tb{G}_1$ and $\tb{G}_2$ with non-negligible advantage.
That is, for any adversary $\mc{A}$, there exists a PPT algorithm $\mc{B}_2$
such that
    $\big| \Adv_{\mc{A}}^{G_{1,k-1}} - \Adv_{\mc{A}}^{G_{1,k}} \big|
    = \Adv_{\mc{B}_2}^{LW2}(\lambda)$.
\end{lemma}

\begin{proof}
%The proof of this lemma is almost same as the proof of Lemma 2 in
%\cite{LewkoW10} except that the proof is employed in the PKS setting.
%
Suppose there exists an adversary $\mc{A}$ that distinguishes between
$\tb{G}_{1,k-1}$ and $\tb{G}_{1,k}$ with non-negligible advantage. A
simulator $\mc{B}_2$ that solves the LW2 assumption using $\mc{A}$ is given:
a challenge tuple
    $D = ((p, \G, \hat{\G}, \G_T, e), \lb
    k, k^a, k^b, k^c, \hat{k}^a, \hat{k}^{a^2}, \hat{k}^{bx}, \hat{k}^{abx},
    \hat{k}^{a^2x})$ and $T$
where $T = T_0 = k^{bc}$ or $T = T_1 = k^{bc+d}$. Then $\mc{B}_2$ that
interacts with $\mc{A}$ is described as follows: $\mc{B}_2$ first selects
random exponents $\nu, y_{\tau}, A, B, \alpha, y_u, y_h, y_w \in \Z_p$. It
computes
    $w_1 = w^{\phi_1} = ((k^b)^{-\nu} k^a k^{y_{\tau}})^{y_w},
    w_2 = w^{\phi_2} = (k^b)^{y_w},
    w = k^{y_w}$
by implicitly setting $\phi_1 = -\nu b + (a + y_{\tau}), \phi_2 = b$. It
implicitly sets $\tau = a + y_{\tau}$ and publishes a public key $PK$ by
selecting random values $c_g, c_u, c_h \in \Z_p$ as
    \begin{align*}
    &   g w_1^{c_g} = k^a w_1^{c_g}, w_2^{c_g}, w^{c_g},~
        u w_1^{c_u} = (k^a)^A k^{y_u} w_1^{c_u}, w_2^{c_u}, w^{c_u},~
        h w_1^{c_h} = (k^a)^B k^{y_h} w_1^{c_h}, w_2^{c_h}, w^{c_h},~
        w_1, w_2, w,~ \\
    &   \hat{g} = \hat{k}^a, \hat{g}^{\nu},
        \hat{g}^{-\tau} = (\hat{k}^{a^2} (\hat{k}^a)^{y_{\tau}})^{-1}),~
        \hat{u} = (\hat{k}^a)^A \hat{k}^{y_u}, \hat{u}^{\nu},
        \hat{u}^{-\tau} = ((\hat{k}^{a^2})^A (\hat{k}^a)^{y_u + Ay_{\tau}}
            \hat{k}^{y_u y_{\tau}})^{-1},~ \\
    &   \hat{h} = (\hat{k}^a)^B \hat{k}^{y_h}, \hat{h}^{\nu},
        \hat{h}^{-\tau} = ((\hat{k}^{a^2})^B (\hat{k}^a)^{y_h + By_{\tau}}
            \hat{k}^{y_h y_{\tau}} )^{-1},~
        \Omega = e(k^a, \hat{k}^a)^{\alpha}.
    \end{align*}
Additionally, it sets $f = k, \hat{f} = \hat{k}$ for the semi-functional
signature and verification. $\mc{A}$ adaptively requests a signature for a
message $M$. If this is a $j$-th signature query, then $\mc{B}_2$ handles
this query as follows:
\begin{itemize}
\item {Case $j < k$} : It creates a semi-functional signature by calling
    \tb{PKS2.SignSF} since it knows the tuple $(f^{-\nu}, f, 1)$ for the
    semi-functional signature.

\item {Case $j = k$} : It selects random exponents $r', c'_1, c'_2 \in
    \Z_p$ and creates a signature by implicitly setting $r = -c + r',~
    c_1 = c(AM+B)/y_w + c'_1,~ c_2 = c/y_w + c'_2$ as
    \begin{align*}
    &   W_{1,1} = g^{\alpha} (k^c)^{-(y_u M + y_h)} (u^M h)^{r'}
                  (T)^{-\nu (AM+B)} (k^c)^{y_{\tau} (AM+B)} w_1^{c'_1} ,~
        W_{1,2} = (T)^{(AM+B)} w_2^{c'_1},~ \\
    &   W_{1,3} = (k^c)^{(AM+B)} w^{c'_1},~
        W_{2,1} = g^{r'} (T)^{-\nu} (k^c)^{y_{\tau}} w_1^{c'_2},~
        W_{2,2} = T w_2^{c'_2},~
        W_{2,3} = k^c w^{c'_2}.
    \end{align*}

\item {Case $j > k$} : It creates a normal signature by calling
    \tb{PKS2.Sign} since it knows the private key.
\end{itemize}

\noindent Finally, $\mc{A}$ outputs a forged signature $\sigma^* =
(W_{1,1}^*, \ldots, W_{2,3}^*)$ on a message $M^*$. To verify the forged
signature, $\mc{B}_2$ first chooses a random exponent $t' \in \Z_p$ and
computes semi-functional verification components by implicitly setting
    $t = bx + t',~ s_c = -a^2 x,~ z_c = A M^* + B$
as
    \begin{align*}
    &   V_{1,1} = \hat{k}^{abx} (\hat{k}^a)^{t'},~
        V_{1,2} = (\hat{k}^{abx})^{\nu} (\hat{k}^a)^{\nu t'} (\hat{k}^{a^2x})^{-1},~
        V_{1,3} = (\hat{k}^{abx})^{-y_{\tau}} (\hat{g}^{-y_{\tau}})^{t'}, \\
    &   V_{2,1} = (\hat{k}^{abx})^{A M^* + B} (\hat{k}^{bx})^{y_u M^* + y_h}
                  (\hat{u}^{M^*} \hat{h})^{t'},~
        V_{2,2} = (\hat{k}^{abx})^{(A M^* + B) \nu}
                  (\hat{k}^{bx})^{(y_u M^* + y_h) \nu}
                  (\hat{u}^{M^*} \hat{h})^{\nu t'},~ \\
    &   V_{2,3} = (\hat{k}^{abx})^{-(A M^* + B) y_{\tau}}
                  (\hat{k}^{abx})^{-(y_u M^* + y_h)}
                  (\hat{k}^{bx})^{-(y_u M^* + y_h) y_{\tau}}
                  ((\hat{u}^{-\tau})^{M^*} \hat{h}^{-\tau})^{t'}.
    \end{align*}
Next, it verifies that
    $\prod_{i=1}^3 e(W_{1,i}^*, V_{1,i}) \cdot
    \prod_{i=1}^3 e(W_{2,i}^*, V_{2,i})^{-1} \stackrel{?}{=}
    e(k^a, \hat{k}^{abx})^{\alpha} \cdot e(k^a, \hat{k}^a)^{\alpha t'}$.
If this equation holds, then it outputs 0. Otherwise, it outputs 1.

\vs To finish the proof, we should show that the distribution of the
simulation is correct. We first show that the distribution of the simulation
using $D, T_0 = k^{bc}$ is the same as $\tb{G}_{1,k-1}$. The public key is
correctly distributed since the random blinding values $y_u, y_h, y_w$ are
used. The $k$-th signature is correctly distributed as
    \begin{align*}
    W_{1,1} &= g^{\alpha} (u^M h)^r w_1^{c_1}
             = g^{\alpha} (k^{(aA+y_u)M} k^{aB+y_h})^{-c + r'}
               (k^{y_w (-\nu b + a + y_{\tau})})^{c (AM+B)/y_w + c'_1} \\
            &= g^{\alpha} (k^c)^{-(y_u M + y_h)} (u^M h)^{r'}
               (T)^{-\nu (AM+B)} (k^c)^{y_{\tau} (AM+B)} w_1^{c'_1},~ \db \\
    W_{1,2} &= w_2^{c_1}
             = (k^{y_w b})^{c (AM+B)/y_w + c'_1}
             = (T)^{(AM+B)} w_2^{c'_1},~ \\
    W_{1,3} &= w^{c_1}
             = (k^{y_w})^{c (AM+B)/y_w + c'_1}
             = (k^c)^{(AM+B)} w^{c'_1}.
    \end{align*}
The semi-functional verification components are correctly distributed as
    \begin{align*}
    V_{2,1} &= (\hat{u}^{M^*} \hat{h})^t
             = (\hat{k}^{(aA+y_u) M^*} \hat{k}^{aB+y_h})^{bx + t'}
             = (\hat{k}^{abx})^{AM^* + B} (\hat{k}^{bx})^{y_u M^* + y_h}
               (\hat{u}^{M^*} \hat{h})^{t'},~ \\
    V_{2,2} &= ((\hat{u}^{\nu})^{M^*} \hat{h}^{\nu})^t \hat{f}^{s_c z_c}
             = (\hat{k}^{(aA+y_u) \nu M^*} \hat{k}^{(aB+y_h) \nu})^{bx + t'}
               \hat{k}^{-a^2x (AM^*+B)} \\
            &= (\hat{k}^{abx})^{(AM^* + B) \nu} (\hat{k}^{bx})^{(y_u M^* + y_h) \nu}
               ((\hat{u}^{\nu})^{M^*} \hat{h}^{\nu})^{t'}
               (\hat{k}^{a^2 x})^{-(AM^* + B)},~
               \db \\
    V_{2,3} &= ((\hat{u}^{-\tau})^{M^*} \hat{h}^{-\tau})^t (\hat{f}^{-\phi_2})^{s_c z_c}
             = (\hat{k}^{-(aA+y_u) (a+y_{\tau}) M^*}
               \hat{k}^{-(aB+y_h) (a+y_{\tau})})^{bx + t'}
               \hat{k}^{-b (-a^2 x) (AM^*+B)} \\
            &= (\hat{k}^{abx})^{-(AM^* + B) y_{\tau} - (y_u M^* + y_h)}
               (\hat{k}^{bx})^{-(y_u M^* + y_h) y_{\tau}}
               ((\hat{u}^{-\tau})^{M^*} \hat{h}^{-\tau})^{t'}.
    \end{align*}
The simulator can create the semi-functional verification components with
only fixed $z_c = AM^* + B$ since $s_c, s_c$ enable the cancellation of
$\hat{k}^{a^2 b x}$. Even though it uses the fixed $z_c$, the distribution of
$z_c$ is correct since $A,B$ are information theoretically hidden to
$\mc{A}$.
We next show that the distribution of the simulation using $D, T_1 =
k^{bc+d}$ is the same as $\tb{G}_{1,k}$. We only consider the distribution of
the $k$-th signature since $T$ is only used in the $k$-th signature. The only
difference between $T_0$ and $T_1$ is that $T_1$ additionally has $k^d$. The
signature components $W_{1,1}, W_{1,2}$, $W_{2,1}, W_{2,2}$ that have $T$ in
the simulation additionally have $(k^d)^{-\nu (AM+B)}$, $(k^d)^{(AM+B)}$,
$(k^d)^{-\nu}$, $k^d$ respectively. If we implicitly set $s_k = d, z_k =
AM+B$, then the distribution of the $k$-th signature is the same as
$\tb{G}_{1,k}$ except that the $k$-th signature is nominally semi-functional.

Finally, we show that $\mc{A}$ cannot distinguish the nominally
semi-functional signature from the semi-functional signature. The main idea
of this is that $\mc{A}$ cannot request a signature for the forgery message
$M^*$ in the security model. Suppose there exists an unbounded adversary,
then he can gather $z_k = AM + B$ from the $k$-th signature and $z_c = AM^* +
B$ from the forged signature. It is easy to show that $z_k, z_c$ look random
to the unbounded adversary since $f(M) = AM + B$ is a pair-wise independent
function and $A, B$ are information theoretically hidden to the adversary.
%
%We obtain $\Pr [\mc{B}_2(D,T_0) = 0] = \Adv_{\mc{A}}^{G_{1,k-1}}$ and $\Pr
%[\mc{B}_2(D,T_1) = 0] = \Adv_{\mc{A}}^{G_{1,k}}$ from the above analysis.
%Thus, we can easily derive the advantage of $\mc{B}_2$ as
%    \begin{eqnarray*}
%    \Adv_{\mc{B}_2}^{LW2}(\lambda)
%     =  \big| \Pr[\mc{B}_2(D, T_0) = 0] - \Pr[\mc{B}_2(D, T_1) = 0] \big|
%     =  \big| \Adv_{\mc{A}}^{G_{1,k-1}} - \Adv_{\mc{A}}^{G_{1,k}} \big|.
%    \end{eqnarray*}
This completes our proof.
\end{proof}

\begin{lemma} \label{lem:pks2-prime-3}
If the DBDH assumption holds, then no polynomial-time adversary can
distinguish between $\tb{G}_2$ and $\tb{G}_3$ with non-negligible advantage.
That is, for any adversary $\mc{A}$, there exists a PPT algorithm $\mc{B}_3$
such that
    $\big| \Adv_{\mc{A}}^{G_2} - \Adv_{\mc{A}}^{G_3} \big| =
    \Adv_{\mc{B}_3}^{DBDH}(\lambda)$.
\end{lemma}

\begin{proof}
%The proof of this lemma is almost same as the proof of Lemma 3 in
%\cite{LewkoW10} except that the proof is employed in the PKS setting.
%
Suppose there exists an adversary $\mc{A}$ that distinguish $\tb{G}_2$ from
$\tb{G}_3$ with non-negligible advantage. A simulator $\mc{B}_3$ that solves
the DBDH assumption using $\mc{A}$ is given: a challenge tuple
    $D = ((p, \G, \hat{\G}, \G_T, e), \lb
    k, k^a, k^b, k^c, \lb \hat{k}, \hat{k}^a, \hat{k}^b, \hat{k}^c)$
    and $T$ where $T = T_0 = e(k, \hat{k})^{abc}$ or $T = T_1 = e(k, \hat{k})^d$.
Then $\mc{B}_3$ that interacts with $\mc{A}$ is described as follows:
$\mc{B}_3$ first chooses random exponents $\phi_1, \phi_2, y_g, x, y \in
\Z_p$ and a random element $w \in \G$. It computes
    $g = k^{y_g}, u = g^{x}, h = g^{y},
    \hat{g} = \hat{k}^{y_g}, \hat{u} = \hat{g}^{x}, \hat{h} = \hat{g}^{y},
    w_1 = w^{\phi_1}, w_2 = w^{\phi_2}$.
It implicitly sets $\nu = a, \tau = \phi_1 + a \phi_2, \alpha = ab$ and
publishes a public key $PK$ by selecting random values $c_g, c_u, c_h \in
\Z_p$ as
    \begin{align*}
    &   g w_1^{c_g}, w_2^{c_g}, w^{c_g},~
        u w_1^{c_u}, w_2^{c_u}, w^{c_u},~
        h w_1^{c_h}, w_2^{c_h}, w^{c_h},~
        w_1, w_2, w,~ \\
    &   \hat{g}, \hat{g}^{\nu} = (\hat{k}^a)^{y_g},
        \hat{g}^{-\tau} = \hat{k}^{-y_g \phi_1} (\hat{k}^a)^{-y_g \phi_2},~
        \hat{u}, \hat{u}^{\nu} = (\hat{g}^{\nu})^x,
        \hat{u}^{-\tau} = (\hat{g}^{-\tau})^x,~ \\
    &   \hat{h}, \hat{h}^{\nu} = (\hat{g}^{\nu})^y,
        \hat{h}^{-\tau} = (\hat{g}^{-\tau})^y,~
        \Omega = e(k^a,\hat{k}^b)^{y_g^2}.
    \end{align*}
Additionally, it sets $f = k, \hat{f} = \hat{k}$ for the semi-functional
signature and semi-functional verification. $\mc{A}$ adaptively requests a
signature for a message $M$. To respond to this query, $\mc{B}_3$ selects
random exponents $r, c_1, c_2, s_k, z'_k \in \Z_p$ and creates a
semi-functional signature by implicitly setting $z_k = b y_g / s_k + z'_k$ as
    \begin{align*}
    &   W_{1,1} = (u^M h)^r w_1^{c_1} (k^a)^{-s_k z'_k},~
        W_{1,2} = w_2^{c_1} (k^b)^{y_g} k^{s_k z'_k},~
        W_{1,3} = w^{c_1}, \\
    &   W_{2,1} = g^r w_1^{c_2} (k^a)^{-s_k},~
        W_{2,2} = w_2^{c_2} k^{s_k},~
        W_{2,3} = w^{c_2}.
    \end{align*}
It can only create a semi-functional signature since $s_k, z_k$ enables the
cancellation of $k^{ab}$. Finally, $\mc{A}$ outputs a forged signature
$\sigma^* = (W_{1,1}^*, \ldots, W_{2,3}^*)$ on a message $M^*$. To verify the
forged signature, $\mc{B}_3$ first chooses random exponents $s_1, s_2, s'_c,
z'_c \in \Z_p$ and computes semi-functional verification components by
implicitly setting
    $t = c,~ s_c = -ac y_g + s'_c,~ z_c = -ac y_g (xM^*+y)/s_c + z'_c/s_c$
as
    \begin{align*}
    &   V_{1,1} = (\hat{k}^c)^{y_g},~
        V_{1,2} = \hat{k}^{s'_c},~
        V_{1,3} = (\hat{k}^c)^{-y_g \phi_1} \hat{k}^{-\phi_2 s'_c}, \\
    &   V_{2,1} = (\hat{k}^c)^{y_g (xM^*+y)},~
        V_{2,2} = \hat{k}^{z'_c},~
        V_{2,3} = (\hat{k}^c)^{-y_g \phi_1 (xM^*+y)} \hat{k}^{-\phi_2 z'_c}.
    \end{align*}
Next, it verifies that $\prod_{i=1}^3 e(W_{1,i}^*, V_{1,i}) \cdot
\prod_{i=1}^3 e(W_{2,i}^*, V_{2,i})^{-1} \stackrel{?}{=} (T)^{y_g^2}.$ If
this equation holds, then it outputs $0$. Otherwise, it outputs $1$.

\vs To finish the proof, we first show that the distribution of the
simulation using $D, T = e(k,\hat{k})^{abc}$ is the same as $\tb{G}_2$. The
public key is correctly distributed since the random values $y_g, x, y, c_g,
c_u, c_h$ are used. The semi-functional signature is correctly distributed as
    \begin{align*}
    W_{1,1} &= g^{\alpha} (u^M h)^r w_1^{c_1} (f^{-\nu})^{s_k z_k}
             = k^{y_g ab} (u^M h)^r w_1^{c_1}
               (k^{-a})^{s_k (b y_g/s_k + z'_k)}
             = (u^M h)^r w_1^{c_1} (k^a)^{-s_k z'_k}.
    \end{align*}
The simulator can only create a semi-functional signature since $z_k = b y_g
/ s_k + z'_k$ enables the cancellation of $k^{ab}$. The semi-functional
verification components are correctly distributed as
    \begin{align*}
    V_{1,1} &= \hat{g}^t = (\hat{k}^{y_g})^c = (\hat{k}^c)^{y_g},~
    V_{1,2}  = (\hat{g}^{\nu})^t \hat{f}^{s_c}
             = (\hat{k}^{y_g a})^c \hat{k}^{-ac y_g + s'_c} = \hat{k}^{s'_c},~ \\
    V_{1,3} &= (\hat{g}^{-\tau})^t (\hat{f}^{-\phi_2})^{s_c}
             = (\hat{k}^{-y_g (\phi_1 + a \phi_2)})^c \hat{k}^{-\phi_2 (-ac y_g + s'_c)}
             = (\hat{k}^c)^{-y_g \phi_1} \hat{k}^{-\phi_2 s'_c}, \db \\
    V_{2,1} &= (\hat{u}^{M^*} \hat{h})^t
             = (\hat{k}^{y_g (xM^*+y)})^c
             = (\hat{k}^c)^{y_g (xM^* + y)}, \\
    V_{2,2} &= (\hat{u}^{\nu M^*} \hat{h}^{\nu})^t \hat{f}^{s_c z_c}
             = (\hat{k}^{y_g a (xM^*+y)})^c \hat{k}^{-ac y_g (xM^*+y) + z'_c}
             = \hat{k}^{z'_c},~ \db \\
    V_{2,3} &= (\hat{u}^{-\tau M^*} \hat{h}^{-\tau})^t (\hat{f}^{-\phi_2})^{s_c z_c}
             = (\hat{k}^{-y_g (\phi_1 + a \phi_2) (xM^*+y)})^c
               (\hat{k}^{-\phi_2})^{-ac y_g (xM^*+y) + z'_c} \\
            &= (\hat{k}^c)^{-y_g \phi_1 (xM^*+y)} \hat{k}^{-\phi_2 z'_c}, \\
    \Omega^t &= e(g,\hat{g})^{\alpha t} = e(k,\hat{k})^{y_g^2 ab c} = (T_0)^{y_g^2} .
    \end{align*}
We next show that the distribution of the simulation using $D, T_1 = e(k,
\hat{k})^d$ is almost the same as $\tb{G}_3$. It is obvious that the
signature verification for the forged signature always fails if $T_1 = e(k,
\hat{k})^d$ is used except with $1/p$ probability since $d$ is a random value
in $\Z_p$.
%
%We obtain $\Pr [\mc{B}_3(D,T_0) = 0] = \Adv_{\mc{A}}^{G_2}$ and $\Pr
%[\mc{B}_3(D,T_1) = 0] = \Adv_{\mc{A}}^{G_3}$ from the above analysis. Thus,
%we can easily derive the advantage of $\mc{B}_3$ as
%    \begin{eqnarray*}
%    \Adv_{\mc{B}_3}^{DBDH}(\lambda)
%     =  \big| \Pr[\mc{B}_3(D, T_0) = 0] - \Pr[\mc{B}_3(D, T_1) = 0] \big|
%     =  \big| \Adv_{\mc{A}}^{G_2} - \Adv_{\mc{A}}^{G_3} \big|.
%    \end{eqnarray*}
This completes our proof.
\end{proof}

\section{Sequential Aggregate Signature} \label{sec:sas}

In this section, we propose two SAS schemes with short public keys and prove
their security based on that of our PKS schemes.

\subsection{Definitions} \label{sec:sas-def}

The concept of SAS was introduced by Lysyanskaya et al.
\cite{LysyanskayaMRS04}. In SAS, all signers first generate public keys and
private keys, and then publishes their public keys. To generate a sequential
aggregate signature, a signer may receive an aggregate-so-far from a previous
signer, and creates a new aggregate signature by adding his signature to the
aggregate-so-far in sequential order. After that, the signer may send the
aggregate signature to a next signer. A verifier can check the validity of
the aggregate signature by using the pubic keys of all signers in the
aggregate signature. An SAS scheme is formally defined as follows:

\begin{definition}[Sequential Aggregate Signature]
A sequential aggregate signature (SAS) scheme consists of four PPT algorithms
\tb{Setup}, \tb{KeyGen}, \tb{AggSign}, and \tb{AggVerify}, which are defined
as follows:
\begin{description}
\item $\tb{Setup}(1^\lambda)$. The setup algorithm takes as input a
    security parameter $1^\lambda$ and outputs public parameters $PP$.

\item $\tb{KeyGen}(PP)$. The key generation algorithm takes as input the
    public parameters $PP$, and outputs a public key $PK$ and a private key
    $SK$.

\item $\tb{AggSign}(AS', \vect{M}, \vect{PK}, M, SK)$. The aggregate
    signing algorithm takes as input an aggregate-so-far $AS'$ on messages
    $\vect{M} = (M_1, \ldots, M_l)$ under public keys $\vect{PK} = (PK_1,
    \ldots, PK_l)$, a message $M$, and a private key $SK$, and outputs a
    new aggregate signature $AS$.

\item $\tb{AggVerify}(AS, \vect{M}, \vect{PK})$. The aggregate verification
    algorithm takes as input an aggregate signature $AS$ on messages
    $\vect{M} = (M_1, \ldots, M_l)$ under public keys $\vect{PK} = (PK_1,
    \ldots, PK_l)$, and outputs either $1$ or $0$ depending on the validity
    of the sequential aggregate signature.
\end{description}
The correctness requirement is that for each $PP$ output by \tb{Setup}, for
all $(PK,SK)$ output by \tb{KeyGen}, any $M$, we have that $\tb{AggVerify}
(\tb{AggSign} (AS', \vect{M}', \vect{PK}', M, SK), \vect{M}'||M,
\vect{PK}'||PK) = 1$ where $AS'$ is a valid aggregate-so-far signature on
messages $\vect{M}'$ under public keys $\vect{PK}'$.
\end{definition}

A trivial SAS scheme can be constructed from a PKS scheme by concatenating
each signer's signature in sequential order, but the size of aggregate
signature is proportional to the size of signers. Therefore, a non-trivial
SAS scheme should satisfy the signature compactness property that requires
the size of aggregate signature to be independent of the size of signers.

The security model of SAS was defined by Lysyanskaya et al.
\cite{LysyanskayaMRS04}, but we follow the security model of Lu et al.
\cite{LuOSSW06} that requires for an adversary to register the key-pairs of
other signers except the target signer, namely the knowledge of secret key
(KOSK) setting or the proof of knowledge (POK) setting. In this security
model, an adversary first given the public key of a target signer. After
that, the adversary adaptively requests a certification for a public key by
registering the key-pair of other signer, and he adaptively requests a
sequential aggregate signature by providing a previous aggregate signature to
the signing oracle. Finally, the adversary outputs a forged sequential
aggregate signature on messages under public keys. If the forged sequential
signature satisfies the conditions of the security model, then the adversary
wins the security game. The security model of SAS is formally defined as
follows:

\begin{definition}[Security]
The security notion of existential unforgeability under a chosen message
attack is defined in terms of the following experiment between a challenger
$\mc{C}$ and a PPT adversary $\mc{A}$:
\begin{enumerate}
\item \tb{Setup}: $\mc{C}$ first initializes a certification list $CL$ as
    empty. Next, it runs \tb{Setup} to obtain public parameters $PP$ and
    \tb{KeyGen} to obtain a key pair $(PK,SK)$, and gives $PK$ to $\mc{A}$.

\item \tb{Certification Query}: $\mc{A}$ adaptively requests the
    certification of a public key by providing a key pair $(PK,SK)$. Then
    $\mc{C}$ adds the key pair $(PK,SK)$ to $CL$ if the key pair is a valid
    one.

\item \tb{Signature Query}: $\mc{A}$ adaptively requests a sequential
    aggregate signature (by providing an aggregate-so-far $AS'$ on messages
    $\vect{M}'$ under public keys $\vect{PK}'$), on a message $M$ to sign
    under the challenge public key $PK$, and receives a sequential
    aggregate signature $AS$.

\item \tb{Output}: Finally (after a sequence of the above queries),
    $\mc{A}$ outputs a forged sequential aggregate signature $AS^*$ on
    messages $\vect{M}^*$ under public keys $\vect{PK}^*$. $\mc{C}$ outputs
    $1$ if the forged signature satisfies the following three conditions,
    or outputs $0$ otherwise: 1) $\tb{AggVerify}(AS^*, \vect{M}^*,
    \vect{PK}^*) = 1$, 2) The challenge public key $PK$ must exists in
    $\vect{PK}^*$ and each public key in $\vect{PK}^*$ except the challenge
    public key must be in $CL$, and 3) The corresponding message $M$ in
    $\vect{M}^*$ of the challenge public key $PK$ must not have been
    queried by $\mc{A}$ to the sequential aggregate signing oracle.
\end{enumerate}
The advantage of $\mc{A}$ is defined as $\Adv_{\mc{A}}^{SAS}(\lambda) = \Pr
[\mc{C} = 1]$ where the probability is taken over all the randomness of the
experiment. An SAS scheme is existentially unforgeable under a chosen message
attack if all PPT adversaries have at most a negligible advantage in the
above experiment.
\end{definition}

\subsection{Construction}

To construct an SAS scheme from a PKS scheme, the PKS scheme should support
multi-users by sharing some elements among all signers and the randomness of
signatures should be sequentially aggregated to a single value. We can employ
the randomness reuse technique of Lu et al. \cite{LuOSSW06} to aggregate the
randomness of signatures. To apply the randomness reuse technique, we should
re-randomize the aggregate signature to prevent a forgery attack. Thus we
build on the PKS schemes of the previous section that support multi-users and
public re-randomization to construct SAS schemes.

\subsubsection{Our SAS1 Scheme}

Our first SAS scheme in prime order bilinear groups is described as follows:

\begin{description}
\item [\textbf{SAS1.Setup}($1^\lambda$):] This algorithm first generates
    the asymmetric bilinear groups $\G, \hat{\G}$ of prime order $p$ of bit
    size $\Theta(\lambda)$. It chooses random elements $g, w \in \G$ and
    $\hat{g}, \hat{v} \in \hat{\G}$. Next, it chooses random exponents
    $\nu_1, \nu_2, \nu_3, \phi_1, \phi_2, \phi_3 \in \Z_p$ and sets $\tau =
    \phi_1 + \nu_1 \phi_2 + \nu_2 \phi_3, \pi = \phi_2 + \nu_3 \phi_3$. It
    also sets $w_1 = w^{\phi_1}, w_2 = w^{\phi_2}, w_3 = w^{\phi_3}$. It
    publishes public parameters as
    \begin{align*}
    PP = \Big(~ (p, \G, \hat{\G}, \G_T, e),~
        g,~ w_1, w_2, w_3, w,~
        \hat{g}, \hat{g}^{\nu_1}, \hat{g}^{\nu_2}, \hat{g}^{-\tau},~
        \hat{v}, \hat{v}^{\nu_3}, \hat{v}^{-\pi}
    ~\Big).
    \end{align*}

\item [\textbf{SAS1.KeyGen}($PP$):] This algorithm takes as input the
    public parameters $PP$. It selects random exponents $\alpha, x, y \in
    \Z_p$ and computes $u = g^x, h = g^y, \hat{u} = \hat{g}^x,
    \hat{u}^{\nu_1} = (\hat{g}^{\nu_1})^x, \hat{u}^{\nu_2} =
    (\hat{g}^{\nu_2})^x, \hat{u}^{-\tau} = (\hat{g}^{-\tau})^x, \hat{h} =
    \hat{g}^y, \hat{h}^{\nu_1} = (\hat{g}^{\nu_1})^y, \hat{h}^{\nu_2} =
    (\hat{g}^{\nu_2})^y, \hat{h}^{-\tau} = (\hat{g}^{-\tau})^y$. It outputs
    a private key $SK = (\alpha, x, y)$ and a public key as
    \begin{align*}
    PK = \Big(~
        u, h,~
        \hat{u}, \hat{u}^{\nu_1}, \hat{u}^{\nu_2}, \hat{u}^{-\tau},~
        \hat{h}, \hat{h}^{\nu_1}, \hat{h}^{\nu_2}, \hat{h}^{-\tau},~
        \Omega = e(g, \hat{g})^{\alpha}
    ~\Big).
    \end{align*}

\item [\textbf{SAS1.AggSign}($AS', \vect{M}', \vect{PK}', M, SK$):] This
    algorithm takes as input an aggregate-so-far $AS' = (S'_{1,1}, \ldots,
    S'_{2,4})$ on messages $\vect{M}' = (M_1, \ldots, M_{l-1})$ under
    public keys $\vect{PK}' = (PK_1, \ldots, PK_{l-1})$ where $PK_i = (u_i,
    h_i, \ldots, \Omega_i)$, a message $M \in \bits^k$ where $k < \lambda$,
    a private key $SK = (\alpha, x, y)$ with $PK = (u, h, \ldots, \Omega)$
    and $PP$. It first checks the validity of $AS'$ by calling
    $\textbf{AggVerify} (AS', \vect{M}', \vect{PK}')$. If $AS'$ is not
    valid, then it halts. If the public key $PK$ of $SK$ does already exist
    in $\vect{PK}'$, then it halts. Next, it creates temporal aggregate
    components by using the randomness of the previous aggregate-so-far as
    \begin{align*}
    &   T_{1,1} = S'_{1,1} \cdot g^{\alpha} (S'_{2,1})^{xM + y},~
        T_{1,2} = S'_{1,2} \cdot (S'_{2,2})^{xM + y},~
        T_{1,3} = S'_{1,3} \cdot (S'_{2,3})^{xM + y},~ \\
    &   T_{1,4} = S'_{1,4} \cdot (S'_{2,4})^{xM + y},~
        T_{2,1} = S'_{2,1},~
        T_{2,2} = S'_{2,2},~
        T_{2,3} = S'_{2,3},~
        T_{2,4} = S'_{2,4}.
    \end{align*}
    Finally, it selects random exponents $r, c_1, c_2 \in \Z_p$ for
    re-randomization and outputs an aggregate signature as
    \begin{align*}
    AS = \Big(~
    &   S_{1,1} = T_{1,1} \cdot \prod_{i=1}^{l-1} (u_i^{M_i} h_i)^r (u^M h)^r
                  w_1^{c_1},
        S_{1,2} = T_{1,2} \cdot w_2^{c_1},
        S_{1,3} = T_{1,3} \cdot w_3^{c_1},
        S_{1,4} = T_{1,4} \cdot w^{c_1},~ \\
    &   S_{2,1} = T_{2,1} \cdot g^r w_1^{c_2},
        S_{2,2} = T_{2,2} \cdot w_2^{c_2},
        S_{2,3} = T_{2,3} \cdot w_3^{c_2},
        S_{2,4} = T_{2,4} \cdot w^{c_2}
    ~\Big).
    \end{align*}

\item [\textbf{SAS1.AggVerify}($AS, \vect{M}, \vect{PK}$):] This algorithm
    takes as input a sequential aggregate signature $AS$ on messages
    $\vect{M} = (M_1, \ldots, M_l)$ under public keys $\vect{PK} = (PK_1,
    \ldots, PK_l)$ where $PK_i = (u_i, h_i, \ldots, \Omega_i)$. It first
    checks that any public key does not appear twice in $\vect{PK}$ and
    that any public key in $\vect{PK}$ has been certified. If these checks
    fail, then it outputs 0. If $l=0$, then it outputs 1 if $S_1 = S_2 =
    1$, 0 otherwise. It chooses random exponents $t, s_1, s_2 \in \Z_p$ and
    computes verification components as
    \begin{align*}
    &   C_{1,1} = \hat{g}^t,~
        C_{1,2} = (\hat{g}^{\nu_1})^t \hat{v}^{s_1},~
        C_{1,3} = (\hat{g}^{\nu_2})^t (\hat{v}^{\nu_3})^{s_1},~
        C_{1,4} = (\hat{g}^{-\tau})^t (\hat{v}^{-\pi})^{s_1}, \\
    &   C_{2,1} = \prod_{i=1}^l (\hat{u}_i^{M_i} \hat{h}_i)^t,~
        C_{2,2} = \prod_{i=1}^l ((\hat{u}_i^{\nu_1})^{M_i} \hat{h}_i^{\nu_1})^t
                  \hat{v}^{s_2},~
        C_{2,3} = \prod_{i=1}^l ((\hat{u}_i^{\nu_2})^{M_i} \hat{h}_i^{\nu_2})^t
                  (\hat{v}^{\nu_3})^{s_2},~ \\
    &   C_{2,4} = \prod_{i=1}^l ((\hat{u}_i^{-\tau})^{M_i} \hat{h}_i^{-\tau})^t
                  (\hat{v}^{-\pi})^{s_2}.
    \end{align*}
    Next, it verifies that $\prod_{i=1}^4 e(S_{1,i}, C_{1,i}) \cdot
    \prod_{i=1}^4 e(S_{2,i}, C_{2,i})^{-1} \stackrel{?}{=} \prod_{i=1}^l
    \Omega_i^t$. If this equation holds, then it outputs $1$. Otherwise, it
    outputs $0$.
\end{description}

The aggregate signature $AS$ is a valid sequential aggregate signature on
messages $\vect{M}'||M$ under public keys $\vect{PK}'||PK$ with randomness
    $\tilde{r} = r' + r,~
    \tilde{c}_1 = c'_1 + c'_2(xM + y) + c_1,~
    \tilde{c}_2 = c'_2 + c_2$
where $r', c'_1, c'_2$ are random values in $AS'$. The sequential aggregate
signature has the following form
    \begin{align*}
    &   S_{1,1} = \prod_{i=1}^l g^{\alpha_i} \prod_{i=1}^l (u_i^{M_i} h_i)^{\tilde{r}}
                  w_1^{\tilde{c}_1},~
        S_{1,2} = w_2^{\tilde{c}_1},~
        S_{1,3} = w_3^{\tilde{c}_1},~
        S_{1,4} = w^{\tilde{c}_1}, \\
    &   S_{2,1} = g^{\tilde{r}} w_1^{\tilde{c}_2},~
        S_{2,2} = w_2^{\tilde{c}_2},~
        S_{2,3} = w_3^{\tilde{c}_2},~
        S_{2,4} = w^{\tilde{c}_2}.
    \end{align*}
%Therefore, the correctness requirement of SAS is easily obtained from the
%above signature form.

\subsubsection{Our SAS2 Scheme}

Our second SAS scheme in prime order bilinear groups is described as follows:

\begin{description}
\item [\tb{SAS2.Setup}($1^\lambda$):] This algorithm first generates the
    asymmetric bilinear groups $\G, \hat{\G}$ of prime order $p$ of bit
    size $\Theta(\lambda)$. It chooses random elements $g, w \in \G$ and
    $\hat{g} \in \hat{\G}$. Next, it selects random exponents $\nu, \phi_1,
    \phi_2 \in \Z_p$ and sets $\tau = \phi_1 + \nu \phi_2$, $w_1 =
    w^{\phi_1}, w_2 = w^{\phi_2}$. It publishes public parameters by
    selecting a random value $c_g \in \Z_p$ as
    \begin{align*}
    PP = \Big(~ (p, \G, \hat{\G}, \G_T, e),~
        g w_1^{c_g}, w_2^{c_g}, w^{c_g},~ w_1, w_2, w,~
        \hat{g}, \hat{g}^{\nu}, \hat{g}^{-\tau},~
        \Lambda = e(g, \hat{g})
    ~\Big).
    \end{align*}

\item [\tb{SAS2.KeyGen}($PP$):] This algorithm takes as input the public
    parameters $PP$. It selects random exponents $\alpha, x, y \in \Z_p$
    and sets $\hat{u} = \hat{g}^x, \hat{h} = \hat{g}^y$. It outputs a
    private key $SK = (\alpha, x, y)$ and a public key by selecting random
    values $c'_u, c'_h \in \Z_p$ as
    \begin{align*}
    PK = \Big(~
    &   u w_1^{c_u} = (g w_1^{c_g})^x w_1^{c'_u},
        w_2^{c_u}   = (w_2^{c_g})^x w_2^{c'_u},
        w^{c_u}     = (w^{c_g})^x w_2^{c'_u},~ \\
    &   h w_1^{c_h} = (g w_1^{c_g})^y w_1^{c'_u},
        w_2^{c_h}   = (w_2^{c_g})^y w_2^{c'_u},
        w^{c_h}     = (w^{c_g})^y w_2^{c'_u},~ \\
    &   \hat{u}, \hat{u}^{\nu} = (\hat{g}^{\nu})^x,
        \hat{u}^{-\tau} = (\hat{g}^{-\tau})^x,~
        \hat{h}, \hat{h}^{\nu} = (\hat{g}^{\nu})^y,
        \hat{h}^{-\tau} = (\hat{g}^{-\tau})^y,~
        \Omega = \Lambda^{\alpha}
    ~\Big).
    \end{align*}

\item [\tb{SAS2.AggSign}($AS', \vect{M}', \vect{PK}', M, SK$):] This
    algorithm takes as input an aggregate-so-far $AS' = (S'_{1,1}, \ldots,
    S'_{2,3})$ on messages $\vect{M}' = (M_1, \ldots, M_{l-1})$ under
    public keys $\vect{PK}' = (PK_1, \ldots, PK_{l-1})$ where $PK_i = (u_i
    w_1^{c_{u,i}}, \ldots, \Omega_i)$, a message $M \in \Z_p$, a private
    key $SK = (\alpha, x, y)$ with $PK = (u w_1^{c_u}, \ldots, \Omega)$ and
    $PP$. It first checks the validity of $AS'$ by calling
    $\tb{SAS.AggVerify}(AS', \vect{M}', \vect{PK}')$. If $AS'$ is not
    valid, then it halts. If the public key $PK$ of $SK$ does already exist
    in $\vect{PK}'$, then it halts. Next, it creates temporal aggregate
    components by using the randomness of the previous aggregate-so-far as
    \begin{align*}
    &   T_{1,1} = S'_{1,1} (g w_1^{c_g})^{\alpha} (S'_{2,1})^{xM + y},~
        T_{1,2} = S'_{1,2} (w_2^{c_g})^{\alpha} (S'_{2,2})^{xM + y},~
        T_{1,3} = S'_{1,3} (w^{c_g})^{\alpha} (S'_{2,3})^{xM + y},~ \\
    &   T_{2,1} = S'_{2,1},~
        T_{2,2} = S'_{2,2},~
        T_{2,3} = S'_{2,3}.
    \end{align*}
    Finally it selects random exponents $r, c_1, c_2 \in \Z_p$ for
    re-randomization and outputs an aggregate signature as
    \begin{align*}
    AS = \Big(~
    &   S_{1,1} = T_{1,1} \cdot \prod_{i=1}^{l}
                  ((u_i w_1^{c_{u,i}})^{M_i} (h_i w_1^{c_{h,i}}))^r w_1^{c_1},~ \\
    &   S_{1,2} = T_{1,2} \cdot \prod_{i=1}^{l}
                  ((w_2^{c_{u,i}})^{M_i} (w_2^{c_{h,i}}))^r w_2^{c_1},~
        S_{1,3} = T_{1,3} \cdot \prod_{i=1}^{l}
                  ((w^{c_{u,i}})^{M_i} (w^{c_{h,i}}))^r w^{c_1},~ \\
    &   S_{2,1} = T_{2,1} \cdot (g w_1^{c_g})^r w_1^{c_2},~
        S_{2,2} = T_{2,2} \cdot (w_2^{c_g})^r w_2^{c_2},~
        S_{2,3} = T_{2,3} \cdot (w^{c_g})^r w^{c_2}
    ~\Big).
    \end{align*}

\item [\tb{SAS2.AggVerify}($AS, \vect{M}, \vect{PK}$):] This algorithm
    takes as input a sequential aggregate signature $AS$ on messages
    $\vect{M} = (M_1, \ldots, M_l)$ under public keys $\vect{PK} = (PK_1,
    \ldots, PK_l)$ where $PK_i = (u_i w_1^{c_{u,i}}, \ldots, \Omega_i)$. It
    first checks that any public key does not appear twice in $\vect{PK}$
    and that any public key in $\vect{PK}$ has been certified. If these
    checks fail, then it outputs 0. If $l=0$, then it outputs 1 if $S_{1,1}
    = \cdots = S_{2,3} = 1$, 0 otherwise. It chooses a random exponent $t
    \in \Z_p$ and computes verification components as
    \begin{align*}
    &   C_{1,1} = \hat{g}^t,~
        C_{1,2} = (\hat{g}^{\nu})^t,~
        C_{1,3} = (\hat{g}^{-\tau})^t, \\
    &   C_{2,1} = \prod_{i=1}^l (\hat{u}_i^{M_i} \hat{h}_i)^t,~
        C_{2,2} = \prod_{i=1}^l ((\hat{u}_i^{\nu})^{M_i} \hat{h}_i^{\nu})^t,~
        C_{2,3} = \prod_{i=1}^l ((\hat{u}_i^{-\tau})^{M_i} \hat{h}_i^{-\tau})^t.
    \end{align*}
    Next, it verifies that $\prod_{i=1}^3 e(S_{1,i}, C_{1,i}) \cdot
    \prod_{i=1}^3 e(S_{2,i}, C_{2,i})^{-1} \stackrel{?}{=} \prod_{i=1}^l
    \Omega_i^t$. If this equation holds, then it outputs $1$. Otherwise, it
    outputs $0$.
\end{description}

Let $r', c'_1, c'_2$ be the randomness of an aggregate-so-far. If we
implicitly sets $\tilde{r} = r' + r,~ \tilde{c}_1 = c'_1 + c_g \alpha_l +
\sum_{i=1}^l (c_{u,i} M_i + c_{h,i}) r + c_1,~ \tilde{c}_2 = c'_2 + c_g r +
c_2$, then the aggregate signature is correctly distributed as
    \begin{align*}
    &   S_{1,1} = \prod_{i=1}^l g^{\alpha_i} \prod_{i=1}^l (u_i^{M_i} h_i)^{\tilde{r}}
                  w_1^{\tilde{c}_1},~
        S_{1,2} = w_2^{\tilde{c}_1},~
        S_{1,3} = w^{\tilde{c}_1}, \\
    &   S_{2,1} = g^{\tilde{r}} w_1^{\tilde{c}_2},~
        S_{2,2} = w_2^{\tilde{c}_2},~
        S_{2,3} = w^{\tilde{c}_2}.
    \end{align*}

\subsection{Security Analysis}

\begin{theorem} \label{thm:sas1-prime}
The above \tb{SAS1} scheme is existentially unforgeable under a chosen
message attack if the \tb{PKS1} scheme is existentially unforgeable under a
chosen message attack. That is, for any PPT adversary $\mc{A}$ for the above
\tb{SAS1} scheme, there exists a PPT algorithm $\mc{B}$ for the \tb{PKS1}
scheme such that
    $\Adv_{\mc{A}}^{SAS}(\lambda) \leq \Adv_{\mc{B}}^{PKS}(\lambda)$.
\end{theorem}

\begin{proof}
Our overall proof strategy for this part follows Lu et al. \cite{LuOSSW06}
and adapts it to our setting. The proof uses two properties: the fact that
the aggregated signature result is independent of the order of aggregation,
and the fact that the simulator of the SAS system possesses the private keys
of all but the target PKS.

Suppose there exists an adversary $\mc{A}$ that forges the above \tb{SAS1}
scheme with non-negligible advantage $\epsilon$. A simulator $\mc{B}$ that
forges the \tb{PKS1} scheme is first given: a challenge public key
    $PK_{PKS} = (
    (p, \G, \hat{\G}, \G_T, e),
    g, u, h, w_1, \ldots, w, \lb
    \hat{g}, \ldots, \hat{g}^{-\tau},
    \hat{u}, \ldots, \hat{u}^{-\tau}, \lb
    \hat{h}, \ldots, \hat{h}^{-\tau},
    \hat{v}, \hat{v}^{\nu_3}, \hat{v}^{-\pi}, \Omega )$.
Then $\mc{B}$ that interacts with $\mc{A}$ is described as follows:
$\mc{B}$ first constructs
    $PP = ( (p, \G, \hat{\G}, \G_T, e),
    g, w_1, \ldots, w, \hat{g}, \ldots, \hat{g}^{-\tau},
    \hat{v}, \hat{v}^{\nu_3}, \hat{v}^{-\pi} )$
and
    $PK^* = ( u, h,
    \hat{u}, \ldots, \hat{u}^{-\tau}, \hat{h}, \ldots, \hat{h}^{-\tau},
    \Omega = e(g, \hat{g})^{\alpha} )$
from $PK_{PKS}$. Next, it initializes a certification list $CL$ as an empty
one and gives $PP$ and $PK^*$ to $\mc{A}$.
$\mc{A}$ may adaptively requests certification queries or sequential
aggregate signature queries. If $\mc{A}$ requests the certification of a
public key by providing a public key $PK_i = (u_i, h_i, \ldots, \Omega_i)$
and its private key $SK_i = (\alpha_i, x_i, y_i)$, then $\mc{B}$ checks the
private key and adds the key pair $(PK_i, SK_i)$ to $CL$.
If $\mc{A}$ requests a sequential aggregate signature by providing an
aggregate-so-far $AS'$ on messages $\vect{M}' = (M_1, \ldots, M_{l-1})$ under
public keys $\vect{PK}' = (PK_1, \ldots, PK_{l-1})$, and a message $M$ to
sign under the challenge private key of $PK^*$, then $\mc{B}$ proceeds the
aggregate signature query as follows:
\begin{enumerate}
\item It first checks that the signature $AS'$ is valid and that each
    public key in $\vect{PK}'$ exits in $CL$.

\item It queries its signing oracle that simulates $\textbf{PKS1.Sign}$ on
    the message $M$ for the challenge public key $PK^*$ and obtains a
    signature $\sigma$.

\item For each $1\leq i\leq l-1$, it constructs an aggregate signature on
    message $M_i$ using $\textbf{SAS1.AggSign}$ since it knows the private
    key that corresponds to $PK_i$. The result signature is an aggregate
    signature for messages $\vect{M}' || M$ under public keys $\vect{PK}'
    || PK^*$ since this scheme does not check the order of aggregation. It
    gives the result signature $AS$ to $\mc{A}$.
\end{enumerate}
Finally, $\mc{A}$ outputs a forged aggregate signature $AS^* = (S_{1,1}^*,
\ldots, S_{2,4}^*)$ on messages $\vect{M}^* = (M_1, \ldots, M_{l})$ under
public keys $\vect{PK}^* = (PK_1, \ldots, PK_l)$ for some $l$. Without loss
of generality, we assume that $PK_1 = PK^*$. $\mc{B}$ proceeds as follows:
\begin{enumerate}
\item $\mc{B}$ first checks the validity of $AS^*$ by calling
    $\textbf{SAS1.AggVerify}$. Additionally, the forged signature should
    not be trivial: the challenge public key $PK^*$ must be in
    $\vect{PK}^*$, and the message $M_1$ must not be queried by $\mc{A}$ to
    the signature query oracle.

\item For each $2 \leq i \leq l$, it parses $PK_i = (u_i, h_i, \ldots,
    \Omega_i)$ from $\vect{PK}^*$, and it retrieves the private key $SK_i =
    (\alpha_i, x_i, y_i)$ of $PK_i$ from $CL$. It then computes
    \begin{align*}
    &   W_{1,1} = S_{1,1}^* \cdot \prod_{i=2}^l
          \big( g^{\alpha_j} (S_{2,1}^*)^{x_i M_i + y_i} \big)^{-1},~
        W_{1,2} = S_{1,2}^* \cdot \prod_{i=2}^l
          \big( (S_{2,2}^*)^{x_i M_i + y_i} \big)^{-1},~
        \displaybreak[0] \\
    &   W_{1,3} = S_{1,3}^* \cdot \prod_{i=2}^l
          \big( (S_{2,3}^*)^{x_i M_i + y_i} \big)^{-1},~
        W_{1,4} = S_{1,4}^* \cdot \prod_{i=2}^l
          \big( (S_{2,4}^*)^{x_i M_i + y_i} \big)^{-1},~ \\
    &   W_{2,1} = S_{2,1}^*,~
        W_{2,2} = S_{2,2}^*,~
        W_{2,3} = S_{2,3}^*,~
        W_{2,4} = S_{2,4}^*.
    \end{align*}

\item It outputs $\sigma = (W_{1,1}, \ldots, W_{2,4})$ as a non-trivial
    forgery of the PKS scheme since it did not make a signing query on
    $M_1$.
\end{enumerate}

To finish the proof, we first show that the distribution of the simulation is
correct. It is obvious that the public parameters and the public key are
correctly distributed. The sequential aggregate signatures is correctly
distributed since this scheme does not check the order of aggregation.
Finally, we can show that the result signature $\sigma = (W_{1,1}, \ldots,
W_{2,4})$ of the simulator is a valid signature for the \tb{PKS1} scheme on
the message $M_1$ under the public key $PK^*$ since it satisfies the
following equation:
    \begin{align*}
    &   \prod_{i=1}^4 e(W_{1,i}, V_{1,i}) \cdot
        \prod_{i=1}^4 e(W_{2,i}, V_{2,i})^{-1} \\
    &=  e(S_{1,1}^*, \hat{g}^t) \cdot
        e(S_{1,2}^*, \hat{g}^{\nu_1 t} \hat{v}^{s_1}) \cdot
        e(S_{1,3}^*, \hat{g}^{\nu_2 t} \hat{v}^{\nu_3 s_1}) \cdot
        e(S_{1,4}^*, \hat{g}^{-\tau t} \hat{v}^{-\pi s_1}) \cdot
        e(\prod_{i=2}^l g^{\alpha_i}, \hat{g}^t)^{-1} \cdot \\
    &~~~~~
        e(S_{2,1}^*, \prod_{i=2}^l (\hat{u}_i^{M_i} \hat{h}_i)^t)^{-1} \cdot
        e(S_{2,2}^*, \prod_{i=2}^l (\hat{u}_i^{M_i} \hat{h}_i)^{\nu_1 t}
            \hat{v}^{\delta_i s_1})^{-1} \cdot
        e(S_{2,3}^*, \prod_{i=2}^l (\hat{u}_i^{M_i} \hat{h}_i)^{\nu_2 t}
            \hat{v}^{\delta_i s_1})^{-1} \cdot \db \\
    &~~~~~
        e(S_{2,4}^*, \prod_{i=2}^l (\hat{u}_i^{M_i} \hat{h}_i)^{-\tau t}
            \hat{v}^{-\pi \delta_i s_1})^{-1} \cdot
        e(S_{2,1}^*, (\hat{u}^{M_1} \hat{h})^t)^{-1} \cdot
        e(S_{2,2}^*, (\hat{u}^{M_1} \hat{h})^{\nu_1 t} \hat{v}^{s_2})^{-1} \cdot
        \\
    &~~~~~
        e(S_{2,3}^*, (\hat{u}^{M_1} \hat{h})^{\nu_2 t} \hat{v}^{\nu_3 s_2})^{-1} \cdot
        e(S_{2,4}^*, (\hat{u}^{M_1} \hat{h})^{-\tau t} \hat{v}^{-\pi s_2})^{-1} \\
    &=  e(S_{1,1}^*, C_{1,1}) \cdot e(S_{1,2}^*, C_{1,2}) \cdot
        e(S_{1,3}^*, C_{1,3}) \cdot e(S_{1,4}^*, C_{1,4}) \cdot
        e(\prod_{i=2}^l g^{\alpha_i}, \hat{g}^t)^{-1} \cdot \db \\
    &~~~~~
        e(S_{2,1}^*, \prod_{i=1}^l (\hat{u}_i^{M_i} \hat{h}_i)^t)^{-1} \cdot
        e(S_{2,2}^*, \prod_{i=1}^l (\hat{u}_i^{M_i} \hat{h}_i)^{\nu_1 t}
            \hat{v}^{\tilde{s}_2})^{-1} \cdot
        e(S_{2,3}^*, \prod_{i=1}^l (\hat{u}_i^{M_i} \hat{h}_i)^{\nu_2 t}
            \hat{v}^{\tilde{s}_2})^{-1} \cdot \db \\
    &~~~~~
        e(S_{2,4}^*, \prod_{i=1}^l (\hat{u}_i^{M_i} \hat{h}_i)^{-\tau t}
            \hat{v}^{-\pi \tilde{s}_2})^{-1} \db \\
    &=  \prod_{i=1}^4 e(S_{1,i}^*, C_{1,i}) \cdot
        \prod_{i=1}^4 e(S_{2,i}^*, C_{2,i})^{-1} \cdot
        e(\prod_{i=2}^l g^{\alpha_i}, \hat{g}^t)^{-1}
     =  \prod_{i=1}^l \Omega_i^t \cdot \prod_{i=2}^l \Omega_i^{-t}
     =  \Omega_1^t
    \end{align*}
where $\delta_i = x_i M_i + y_i$ and $\tilde{s}_2 = \sum_{i=2}^l (x_i M_i +
y_i) s_1 + s_2$. This completes our proof.
\end{proof}

\begin{theorem} \label{thm:sas2-prime}
The above \tb{SAS2} scheme is existentially unforgeable under a chosen
message attack if the \tb{PKS2} scheme is existentially unforgeable under a
chosen message attack. That is, for any PPT adversary $\mc{A}$ for the above
\tb{SAS2} scheme, there exists a PPT algorithm $\mc{B}$ for the \tb{PKS2}
scheme such that
    $\Adv_{\mc{A}}^{SAS}(\lambda) \leq \Adv_{\mc{B}}^{PKS}(\lambda)$.
\end{theorem}

\begin{proof}
%The security proof of this theorem is similar to that of Lu et al.
%\cite{LuOSSW06}. That is, a simulator can simulate the generation of
%aggregate signatures since the verification algorithm of aggregate signatures
%does not check the order of aggregation and the simulator knows the private
%keys of other signers except the target signer. To extract a forged PKS
%signature from the forged aggregate signature, the simulator uses the
%knowledge of private keys of other signers.
%
Suppose there exists an adversary $\mc{A}$ that forges the above \tb{SAS2}
scheme with non-negligible advantage $\epsilon$. A simulator $\mc{B}$ that
forges the \tb{PKS2} scheme is first given: a challenge public key
    $PK_{PKS} = (
    (p, \G, \hat{\G}, \G_T, e),
    g w_1^{c_g}, w_2^{c_g}, w^{c_g}, u w_1^{c_u}, \lb \ldots, w^{c_h},
    w_1, w_2, w, \hat{g}, \hat{g}^{\nu}, \hat{g}^{-\tau}, \hat{u}, \ldots,
    \hat{h}^{-\tau}, \Omega )$.
Then $\mc{B}$ that interacts with $\mc{A}$ is described as follows:
$\mc{B}$ first constructs
    $PP = ( (p, \G, \hat{\G}, \G_T, e),
    g w_1^{c_g}, w_2^{c_g}, w^{c_g}, w_1, w_2, w,
    \hat{g}, \hat{g}^{\nu}, \hat{g}^{-\tau}, \Lambda )$
by computing $\Lambda = e(gw_1^{c_g}, \hat{g}) \cdot e(w_2^{c_g},
\hat{g}^{\nu}) \cdot e(w^{c_g}, \hat{g}^{-\tau}) = e(g, \hat{g})$ and
    $PK^* = ( u w_1^{c_u}, \ldots, w^{c_h}, \hat{u}, \ldots, \hat{h}^{-\tau},
    \Omega )$
from $PK_{PKS}$. Next, it initializes a certification list $CL$ as an empty
one and gives $PP$ and $PK^*$ to $\mc{A}$.
$\mc{A}$ may adaptively requests certification queries or sequential
aggregate signature queries. If $\mc{A}$ requests the certification of a
public key by providing a public key $PK_i = (u_i w_1^{c_{u,i}}, \ldots,
\Omega_i)$ and its private key $SK_i = (\alpha_i, x_i, y_i)$, then $\mc{B}$
checks the private key and adds the key pair $(PK_i, SK_i)$ to $CL$.
If $\mc{A}$ requests a sequential aggregate signature by providing an
aggregate-so-far $AS'$ on messages $\vect{M}' = (M_1, \ldots, M_{l-1})$ under
public keys $\vect{PK}' = (PK_1, \ldots, PK_{l-1})$, and a message $M$ to
sign under the challenge private key of $PK^*$, then $\mc{B}$ proceeds the
aggregate signature query as follows:
\begin{enumerate}
\item It first checks that the signature $AS'$ is valid and that each
    public key in $\vect{PK}'$ exits in $CL$.

\item It queries its signing oracle that simulates $\tb{PKS2.Sign}$ on
    the message $M$ for the challenge public key $PK^*$ and obtains a
    signature $\sigma$.

\item For each $1\leq i\leq l-1$, it constructs an aggregate signature on
    message $M_i$ using $\tb{SAS2.AggSign}$ since it knows the private
    key that corresponds to $PK_i$. The result signature is an aggregate
    signature for messages $\vect{M}' || M$ under public keys $\vect{PK}'
    || PK^*$ since this scheme does not check the order of aggregation.
    It gives the result signature $AS$ to $\mc{A}$.
\end{enumerate}
Finally, $\mc{A}$ outputs a forged aggregate signature $AS^* = (S_{1,1}^*,
\ldots, S_{2,3}^*)$ on messages $\vect{M}^* = (M_1, \ldots, M_{l})$ under
public keys $\vect{PK}^* = (PK_1, \ldots, PK_l)$ for some $l$. Without loss
of generality, we assume that $PK_1 = PK^*$. $\mc{B}$ proceeds as follows:
\begin{enumerate}
\item $\mc{B}$ first checks the validity of $AS^*$ by using
    $\tb{SAS2.AggVerify}$. Additionally, the forged signature should not
    be trivial: the challenge public key $PK^*$ must be in $\vect{PK}^*$,
    and the message $M_1$ must not be queried by $\mc{A}$ to the
    signature query oracle.

\item For each $2 \leq i \leq l$, it parses $PK_i = (u_i w_1^{c_{u,i}},
    \ldots, \Omega_i)$ from $\vect{PK}^*$, and it retrieves the private
    key $SK_i = (\alpha_i, x_i, y_i)$ of $PK_i$ from $CL$. It then
    computes
    \begin{align*}
    &   W_{1,1} = S_{1,1}^* \prod_{i=2}^l \big(
            g^{\alpha_j} (S_{2,1}^*)^{x_i M_i + y_i} \big)^{-1},~
        W_{1,2} = S_{1,2}^* \prod_{i=2}^l \big(
            (S_{2,2}^*)^{x_i M_i + y_i} \big)^{-1},~
        W_{1,3} = S_{1,3}^* \prod_{i=2}^l \big(
            (S_{2,3}^*)^{x_i M_i + y_i} \big)^{-1},~ \\
    &   W_{2,1} = S_{2,1}^*,~
        W_{2,2} = S_{2,2}^*,~
        W_{2,3} = S_{2,3}^*.
    \end{align*}

\item It outputs $\sigma = (W_{1,1}, \ldots, W_{2,3})$ as a non-trivial
    forgery of the PKS scheme since it did not make a signing query on
    $M_1$.
\end{enumerate}

The public parameters and the public key are correctly distributed, and the
sequential aggregate signatures are also correctly distributed since this
scheme does not check the order of aggregation. The result signature $\sigma
= (W_{1,1}, \ldots, W_{2,3})$ of the simulator is a valid PKS signature on
the message $M_1$ under the public key $PK^*$ since it satisfies the
following equation:
    \begin{align*}
    \lefteqn{
        \prod_{i=1}^3 e(W_{1,i}, V_{1,i}) \cdot
        \prod_{i=1}^3 e(W_{2,i}, V_{2,i})^{-1} } \\
    &=  e(S_{1,1}^*, \hat{g}^t) \cdot
        e(S_{1,2}^*, \hat{g}^{\nu t}) \cdot
        e(S_{1,4}^*, \hat{g}^{-\tau t}) \cdot
        e(\prod_{i=2}^l g^{\alpha_i}, \hat{g}^t)^{-1} \cdot
        \db \\
    &\quad
        e(S_{2,1}^*, \prod_{i=2}^l (\hat{u}_i^{M_i} \hat{h}_i)^t)^{-1} \cdot
        e(S_{2,2}^*, \prod_{i=2}^l (\hat{u}_i^{M_i} \hat{h}_i)^{\nu t})^{-1} \cdot
        e(S_{2,3}^*, \prod_{i=2}^l (\hat{u}_i^{M_i} \hat{h}_i)^{-\tau t})^{-1} \cdot \\
    &\quad
        e(S_{2,1}^*, (\hat{u}^{M_1} \hat{h})^t)^{-1} \cdot
        e(S_{2,2}^*, (\hat{u}^{M_1} \hat{h})^{\nu t})^{-1} \cdot
        e(S_{2,3}^*, (\hat{u}^{M_1} \hat{h})^{-\tau t})^{-1}
        \db \\
    &=  e(S_{1,1}^*, C_{1,1}) \cdot e(S_{1,2}^*, C_{1,2}) \cdot
        e(S_{1,3}^*, C_{1,3}) \cdot
        e(\prod_{i=2}^l g^{\alpha_i}, \hat{g}^t)^{-1} \cdot \\
    &\quad
        e(S_{2,1}^*, \prod_{i=1}^l (\hat{u}_i^{M_i} \hat{h}_i)^t)^{-1} \cdot
        e(S_{2,2}^*, \prod_{i=1}^l (\hat{u}_i^{M_i} \hat{h}_i)^{\nu t})^{-1} \cdot
        e(S_{2,3}^*, \prod_{i=1}^l (\hat{u}_i^{M_i} \hat{h}_i)^{-\tau t})^{-1}
        \db \\
    &=  \prod_{i=1}^3 e(S_{1,i}^*, C_{1,i}) \cdot
        \prod_{i=1}^3 e(S_{2,i}^*, C_{2,i})^{-1} \cdot
        e(\prod_{i=2}^l g^{\alpha_i}, \hat{g}^t)^{-1}
     =  \prod_{i=1}^l \Omega_i^t \cdot \prod_{i=2}^l \Omega_i^{-t}
     =  \Omega_1^t
    \end{align*}
where $\delta_i = x_i M_i + y_i$ and $\tilde{s}_2 = \sum_{i=2}^l (x_i M_i +
y_i) s_1 + s_2$. This completes our proof.
\end{proof}

\subsection{Discussions}

\noindent \textbf{Multiple Messages}. The SAS schemes of this paper only
allow a signer to sign once in the aggregate algorithm. To support multiple
signing per one signer, we can use the method of Lu et al. \cite{LuOSSW06}.
The basic idea of Lu et al. is to apply a collision resistant hash function
$H$ to a message $M$ before performing the signing algorithm. If a signer
wants to add a signature on a message $M_2$ into the aggregate signature, he
first removes his previous signature on $H(M_1)$ from the aggregate signature
using his private key, and then he adds the new signature on the
$H(M_1||M_2)$ to the aggregate signature.

\section{Multi-Signature} \label{sec:ms}

In this section, we propose an efficient multi-signature (MS) scheme with
short public parameters and prove its security without random oracles.

\subsection{Definitions}

Multi-Signature (MS) can be regarded as a special kind of PKAS in which
different signatures generated by different signers on the same message are
combined as a short multi-signature. Thus MS consists of four algorithms of
PKS and additional two algorithms \tb{Combine} and \tb{MultiVerify} for
combining a multi-signature and verifying a multi-signature. In MS, each
signer generates a public key and a private key, and he can generate an
individual signature on a message by using his private key. To generate a
multi-signature, anyone can combine individual signatures of different
signers on the same message. A verifier can check the validity of the
multi-signature by using the public keys of signers. An MS scheme is formally
defined as follows:

\begin{definition}[Multi-Signature]
A multi-signature (MS) scheme consists of six PPT algorithms \tb{Setup},
\tb{KeyGen}, \tb{Sign}, \tb{Verify}, \tb{Combine}, and \tb{MultVerify}, which
are defined as follows:
\begin{description}
\item \tb{Setup}($1^\lambda$): The setup algorithm takes as input a
    security parameter $\lambda$, and outputs public parameters $PP$.

\item \tb{KeyGen}($PP$): The key generation algorithm takes as input the
    public parameters $PP$, and outputs a public key $PK$ and a private key
    $SK$.

\item \tb{Sign}($M, SK$): The signing algorithm takes as input a message
    $M$, and a private key $SK$. It outputs a signature $\sigma$.

\item \tb{Verify}($\sigma, M, PK$): The verification algorithm takes as
    input a signature $\sigma$ on a message $M$ under a public key $PK$,
    and outputs either $1$ or $0$ depending on the validity of the
    signature.

\item \tb{Combine}($\vect{\sigma}, M, \vect{PK}$): The combining algorithm
    takes as input signatures $\vect{\sigma}$ on a message $M$ under public
    keys $\vect{PK} = (PK_1, \ldots, PK_l)$, and outputs a multi-signature
    $MS$.

\item \tb{MultVerify}($MS, M, \vect{PK}$): The multi-verification algorithm
    takes as input a multi-signature $MS$ on a message $M$ under public
    keys $\vect{PK} = (PK_1, \ldots, PK_l)$, and outputs either $1$ or $0$
    depending on the validity of the multi-signature.
\end{description}
The correctness requirement is that for each $PP$ output by
$\tb{Setup}(1^\lambda)$, for all $(PK,SK)$ output by $\tb{KeyGen}(PP)$, and
any $M$, we have that $\tb{Verify} (\tb{Sign}(M, SK), M, PK) = 1$ and for
each $\vect{\sigma}$ on message $M$ under public keys $\vect{PK}$,
$\tb{MultVerify} (\tb{Combine} (\vect{\sigma}, M, \vect{PK}), M, \vect{PK}) =
1$.
\end{definition}

The security model of MS was defined by Micali et al. \cite{MicaliOR01}, but
we follow the security model of Boldyreva \cite{Boldyreva03} that requires
for an adversary to register the key-pairs of other signers except the target
signer, namely the knowledge of secret key (KOSK) setting or the proof of
knowledge (POK) setting. In this security model, an adversary is first given
the public key of a target signer. After that, the adversary adaptively
requests the certification of a public key by registering the key-pair of
other signer, and he adaptively requests a signature for the target signer on
a message. Finally, the adversary outputs a forged multi-signature on a
message $M^*$ under public keys. If the forged multi-signature satisfies the
conditions of the security model, then the adversary wins the security game.
The security model of MS is formally defined as follows:

\begin{definition}[Security]
The security notion of existential unforgeability under a chosen message
attack is defined in terms of the following experiment between a challenger
$\mc{C}$ and a PPT adversary $\mc{A}$:
\begin{enumerate}
\item \tb{Setup}: $\mc{C}$ first initialize the certification list $CL$ as
    empty. Next, it runs \tb{Setup} to obtain public parameters $PP$ and
    \tb{KeyGen} to obtain a key pair $(PK,SK)$, and gives $PP, PK$ to
    $\mc{A}$.

\item \tb{Certification Query}: $\mc{A}$ adaptively requests the
    certification of a public key by providing a key pair $(PK,SK)$.
    $\mc{C}$ adds the key pair $(PK,SK)$ to $CL$ if the private key is a
    valid one.

\item \tb{Signature Query}: $\mc{A}$ adaptively requests a signature by
    providing a message $M$ to sign under the challenge public key $PK$,
    and receives a signature $\sigma$.

\item \tb{Output}: Finally, $\mc{A}$ outputs a forged multi-signature
    $MS^*$ on a message $M^*$ under public keys $\vect{PK}^*$. $\mc{C}$
    outputs $1$ if the forged signature satisfies the following three
    conditions, or outputs $0$ otherwise: 1) $\tb{MultVerify} (MS^*, M^*,
    \vect{PK}^*) = 1$, 2) The challenge public key $PK$ must exists in
    $\vect{PK}^*$ and each public key in $\vect{PK}^*$ except the challenge
    public key must be in $CL$, and 3) The message $M^*$ must not have been
    queried by $\mc{A}$ to the signing oracle.
\end{enumerate}
The advantage of $\mc{A}$ is defined as $\Adv_{\mc{A}}^{MS} = \Pr [\mc{C} =
1]$ where the probability is taken over all the randomness of the experiment.
An MS scheme is existentially unforgeable under a chosen message attack if
all PPT adversaries have at most a negligible advantage in the above
experiment.
\end{definition}

\subsection{Construction}

To construct an MS scheme with short public parameters, we may use our PKS
schemes that support multi-users and public re-randomization. To aggregate
the randomness of signatures, we cannot use the technique of Lu et al.
\cite{LuOSSW06} since the randomness should be freely aggregated in MS.
Instead we aggregate the randomness of signatures by using the fact that each
signer generates a signature on the same message in MS. That is, if group
elements $u, h$ that are related to message hashing are shared among all
signers, then the randomness of each signer can be easily aggregated since
the random exponent in a public key and the randomness of a signature are
placed in different positions. Thus our two PKS schemes can be used to build
MS schemes since $g, u, h$ in PKS1 or $g w_1^{c_g}, u w_1^{c_u}, h w_1^{c_h}$
in PKS2 are published in a public key. Note that it is not required for a
signer to publicly re-randomize a multi-signature since each signer selects
an independent random value.

To reduce the size of multi-signatures, we use our PKS2 scheme for this MS
scheme. Our MS scheme based on the PKS2 scheme is described as follows:

\begin{description}
\item [\tb{MS.Setup}($1^\lambda$):] This algorithm first generates the
    asymmetric bilinear groups $\G, \hat{\G}$ of prime order $p$ of bit
    size $\Theta(\lambda)$. It chooses random elements $g, w \in \G$ and
    $\hat{g} \in \hat{\G}$. Next, it selects random exponents $\nu, \phi_1,
    \phi_2 \in \Z_p$ and sets $\tau = \phi_1 + \nu \phi_2$, $w_1 =
    w^{\phi_1}, w_2 = w^{\phi_2}$. It selects random exponents $x, y \in
    \Z_n$ and computes $u = g^x, h = g^y, \hat{u} = \hat{g}^x, \hat{h} =
    \hat{g}^y$. It publishes public parameters by selecting random values
    $c_g, c_u, c_h \in \Z_p$ as
    \begin{align*}
    PP = \Big(~
    &   (p, \G, \hat{\G}, \G_T, e),~
        g w_1^{c_g}, w_2^{c_g}, w^{c_g},~
        u w_1^{c_u}, w_2^{c_u}, w^{c_u},~
        h w_1^{c_h}, w_2^{c_h}, w^{c_h},~ \\
    &   w_1, w_2, w,~
        \hat{g}, \hat{g}^{\nu}, \hat{g}^{-\tau},~
        \hat{u}, \hat{u}^{\nu}, \hat{u}^{-\tau},~
        \hat{h}, \hat{h}^{\nu}, \hat{h}^{-\tau},~
        \Lambda = e(g, \hat{g})
    ~\Big).
    \end{align*}

\item [\tb{MS.KeyGen}($PP$):] This algorithm takes as input the public
    parameters $PP$. It selects a random exponent $\alpha \in \Z_p$ and
    computes $\Omega = \Lambda^{\alpha}$. Then it outputs a private key $SK
    = \alpha$ and a public key as $PK = \Omega$.

\item [\tb{MS.Sign}($M, SK$):] This algorithm takes as input a message $M
    \in \Z_p$ and a private key $SK = \alpha$. It selects random exponents
    $r, c_1, c_2 \in \Z_p$ and outputs a signature as
    \begin{align*}
    \sigma = \Big(~
    &   W_{1,1} = (g w_1^{c_g})^{\alpha} ((u w_1^{c_u})^M (h w_1^{c_h}))^r
                  w_1^{c_1},~ \\
    &   W_{1,2} = (w_2^{c_g})^{\alpha} ((w_2^{c_u})^M w_2^{c_h})^r w_2^{c_1},~
        W_{1,3} = (w^{c_g})^{\alpha} ((w^{c_u})^M w^{c_h})^r w^{c_1},~ \\
    &   W_{2,1} = (g w_1^{c_g})^r w_1^{c_2},~
        W_{2,2} = (w_2^{c_g})^r w_2^{c_2},~
        W_{2,3} = (w^{c_g})^r w^{c_2}
    ~\Big).
    \end{align*}

\item [\tb{MS.Verify}($\sigma, M, PK$):] This algorithm takes as input a
    signature $\sigma$ on a message $M$ under a public key $PK$. It chooses
    a random exponent $t \in \Z_p$ and computes verification components as
    \begin{align*}
    &   V_{1,1} = \hat{g}^t,
        V_{1,2} = (\hat{g}^{\nu})^t,
        V_{1,3} = (\hat{g}^{-\tau})^t, \\
    &   V_{2,1} = (\hat{u}^M \hat{h})^t,
        V_{2,2} = ((\hat{u}^{\nu})^M \hat{h}^{\nu})^t,
        V_{2,3} = ((\hat{u}^{-\tau})^M \hat{h}^{-\tau})^t.
    \end{align*}
    Next, it verifies that $\prod_{i=1}^3 e(W_{1,i}, V_{1,i}) \cdot
    \prod_{i=1}^3 e(W_{2,i}, V_{2,i})^{-1} \stackrel{?}{=} \Omega^t$. If
    this equation holds, then it outputs $1$. Otherwise, it outputs $0$.

\item [\tb{MS.Combine}($\vect{\sigma}, M, \vect{PK}$):] This algorithm
    takes as input signatures $\vect{\sigma} = (\sigma_1, \ldots,
    \sigma_l)$ on a message $M$ under public keys $\vect{PK} = (PK_1,
    \ldots, PK_l)$ where $PK_i = \Omega_i$. It first checks the validity of
    each signature $\sigma_i = (W_{1,1}^i, \ldots, W_{2,3}^i)$ by calling
    $\tb{MS.Verify}(\sigma_i, M, PK_i)$. If any signature is invalid, then
    it halts. It then outputs a multi-signature for a message $M$ as
    \begin{align*}
    MS = \Big(~
    &   S_{1,1} = \prod_{i=1}^l W_{1,1}^i,~
        S_{1,2} = \prod_{i=1}^l W_{1,2}^i,~
        S_{1,3} = \prod_{i=1}^l W_{1,3}^i,~ \\
    &   S_{2,1} = \prod_{i=1}^l W_{2,1}^i,~
        S_{2,2} = \prod_{i=1}^l W_{2,2}^i,~
        S_{2,3} = \prod_{i=1}^l W_{2,3}^i
    ~\Big).
    \end{align*}

\item [\tb{MS.MultVerify}($MS, M, \vect{PK}$):] This algorithm takes as
    input a multi-signature $MS$ on a message $M$ under public keys
    $\vect{PK} = (PK_1, \ldots, PK_l)$ where $PK_i = \Omega_i$. It chooses
    a random exponent $t \in \Z_p$ and computes verification components as
    \begin{align*}
    &   V_{1,1} = \hat{g}^t,
        V_{1,2} = (\hat{g}^{\nu})^t,
        V_{1,3} = (\hat{g}^{-\tau})^t, \\
    &   V_{2,1} = (\hat{u}^M \hat{h})^t,
        V_{2,2} = ((\hat{u}^{\nu})^M \hat{h}^{\nu})^t,
        V_{2,3} = ((\hat{u}^{-\tau})^M \hat{h}^{-\tau})^t.
    \end{align*}
    Next, it verifies that $\prod_{i=1}^3 e(S_{1,i}, V_{1,i}) \cdot
    \prod_{i=1}^3 e(S_{2,i}, V_{2,i})^{-1} \stackrel{?}{=} \prod_{i=1}^l
    \Omega_i^t$. If this equation holds, then it outputs $1$. Otherwise, it
    outputs $0$.
\end{description}

\subsection{Security Analysis}

\begin{theorem} \label{thm:ms-prime}
The above \tb{MS} scheme is existentially unforgeable under a chosen message
attack if the \tb{PKS2} scheme is existentially unforgeable under a chosen
message attack. That is, for any PPT adversary $\mc{A}$ for the above \tb{MS}
scheme, there exists a PPT algorithm $\mc{B}$ for the \tb{PKS2} scheme such
that
    $\Adv_{\mc{A}}^{MS}(\lambda) \leq \Adv_{\mc{B}}^{PKS}(\lambda)$.
\end{theorem}

\begin{proof}
Suppose there exists an adversary $\mc{A}$ that forges the above \tb{MS}
scheme with a non-negligible advantage $\epsilon$. A simulator $\mc{B}$ that
forges the \tb{PKS2} scheme is given: a challenge public key $PK_{PKS} = (
(p, \G, \hat{\G}, \G_T, e), \lb g w_1^{c_g}, \ldots, \Lambda, \Omega )$. Then
$\mc{B}$ that interacts with $\mc{A}$ is described as follows:
$\mc{B}$ first constructs $PP = ( (p, \G, \hat{\G}, \G_T, e), \lb g
w_1^{c_g}, \ldots, \Lambda )$ by computing $\Lambda = e(gw_1^{c_g}, \hat{g})
\cdot e(w_2^{c_g}, \hat{g}^{\nu}) \cdot e(w^{c_g}, \hat{g}^{-\tau}) = e(g,
\hat{g})$ and $PK^* = \Omega$ from $PK_{PKS}$. Next, it initialize a
certification list $CL$ as an empty one and gives $PP$ and $PK^*$ to
$\mc{A}$.
$\mc{A}$ may adaptively request certification queries or signature queries.
If $\mc{A}$ requests the certification of a public key by providing a public
key $PK_i = \Omega_i$ and its private key $SK_i = \alpha_i$, then $\mc{B}$
checks the key pair and adds $(PK_i, SK_i)$ to $CL$.
If $\mc{A}$ requests a signature by providing a message $M$ to sign under the
challenge private key of $PK^*$, then $\mc{B}$ queries its signing oracle
that simulates $\tb{PKS2.Sign}$ on the message $M$ for the challenge public
key $PK^*$, and gives the signature to $\mc{A}$.
Finally, $\mc{A}$ outputs a forged multi-signature $MS^* = (S_{1,1}^*,
\ldots, S_{2,3}^*)$ on a message $M^*$ under public keys $\vect{PK}^* =
(PK_1, \ldots, PK_l)$ for some $l$. Without loss of generality, we assume
that $PK_1 = PK^*$. $\mc{B}$ proceeds as follows:
\begin{enumerate}
\item $\mc{B}$ first check the validity of $MS^*$ by calling
    $\tb{MS.MultVerify}$. Additionally, the forged signature should not be
    trivial: the challenge public key $PK^*$ must be in $\vect{PK}^*$, and
    the message $M$ must not be queried by $\mc{A}$ to the signing oracle.

\item For each $2\leq i\leq l$, it parses $PK_i = \Omega_i$ from
    $\vect{PK}^*$, and it retrieves the private key $SK_i = g^{\alpha_i}$
    of $PK_i$ from $CL$. It then computes
    \begin{align*}
    &   W_{1,1} = S_{1,1}^* \cdot \prod_{i=2}^l \big( g^{\alpha_i} \big)^{-1},~
        W_{1,2} = S_{1,2}^*,~
        W_{1,3} = S_{1,3}^*,~ \\
    &   W_{2,1} = S_{2,1}^*,~
        W_{2,2} = S_{2,2}^*,~
        W_{2,3} = S_{2,3}^*.
    \end{align*}

\item It outputs $\sigma = (W_{1,1}, \ldots, W_{2,3})$ as a non-trivial
    forgery of the PKS scheme since it did not make a signing query on
    $M_1$.
\end{enumerate}

To finish the proof, we first show that the distribution of the simulation is
correct. It is obvious that the public parameters, the public key, and the
signatures are correctly distributed. Next we show that the output signature
$\sigma = (W_{1,1}, \ldots, W_{2,3})$ of the simulator is a valid signature
for the \tb{PKS2} scheme on the message $M_1$ under the public key $PK^*$
since it satisfies the following equation
    \begin{align*}
    \lefteqn{
        \prod_{i=1}^3 e(W_{1,i}, V_{1,i}) \cdot
        \prod_{i=1}^3 e(W_{2,i}, V_{2,i})^{-1} } \\
    &=  \prod_{i=1}^3 e(S_{1,i}^*, V_{1,i}) \cdot
        \prod_{i=1}^3 e(S_{2,i}^*, V_{2,i})^{-1} \cdot
        e(\prod_{i=2}^l g^{\alpha_i}, \hat{g})^{-1}
     =  \prod_{i=1}^l \Omega_i^t \cdot
        \prod_{i=2}^l \Omega_i^{-t}
     =  \Omega_1^t.
    \end{align*}
This completes our proof.
\end{proof}

\subsection{Discussions}

\tb{Removing the Proof of Knowledge.} In our MS scheme, an adversary should
prove that he knows the private key of other signer by using a zero-knowledge
proof system. Ristenpart and Yilek \cite{RistenpartY07} showed that some MS
schemes can be proven in the proof of possession (POP) setting instead of the
POK setting. Our MS scheme also can be proven in the POP setting by using
their technique. That is, if our MS scheme is incorporated with a POP scheme
that uses a different hash function, and the adversary submits a signature on
the private key of other signer as the proof of possession, then the security
of our scheme is also achieved. In the security proof, a simulator cannot
extract the private key element $g^{\alpha}$ from the signature of the POP
scheme, but he can extract other values $g^{\alpha} w_1^{c'}, w_2^{c'},
w^{c'}$ and these values are enough for the security proof.

\section{Conclusion}

In this paper, we first proposed two PKS schemes with short public keys that
support multi-users and public re-randomization based on the LW-IBE scheme.
Next, we proposed two SAS schemes with short public keys without random
oracles and with no relaxation of assumptions (i.e., employing neither random
oracles nor interactive assumptions) based on our two PKS schemes. The
proposed SAS schemes are the first of this kind that have short (a constant
number of group elements) size public keys and a constant number of pairing
operations per message in the verification algorithm. We also proposed an MS
scheme with short public parameters based on our PKS scheme and proved its
security without random oracles.

There are many interesting open problems. The first one is to construct an
SAS scheme with short public keys that is secure under standard assumptions
without random oracles. A possible approach is to build an SAS scheme based
on the practical PKS scheme of B{\"{o}}hl et al. \cite{BohlHJKSS13} that is
secure under the standard assumption. The second one is to build an SAS
scheme with short public keys that supports lazy verification and has the
constant size of aggregate signatures. Brogle et al. \cite{BrogleGR12}
proposed an SAS scheme with lazy verification, but the size of aggregate
signatures in their SAS scheme is not constant.

%\section*{Acknowledgements}
%
%This work was partly supported by the MSIP (Ministry of Science, ICT and
%Future Planning), Korea, under the ITRC (Information Technology Research
%Center) support program (NIPA-2014-H0301-14-1004) supervised by the NIPA
%(National IT Industry Promotion Agency) and the National Research Foundation
%of Korea (NRF) grant funded by the Korea government (MEST) (2010-0029121).

\bibliographystyle{plain}
\bibliography{sas-short-pkey-merge}

\appendix

\section{Lewko-Waters IBE} \label{sec:lw-ibe}

In this section, we describe the IBE scheme of Lewko and Waters (LW-IBE)
\cite{LewkoW10} in prime order bilinear groups and the PKS scheme (LW-PKS)
that is derived from the LW-IBE scheme.

\subsection{The LW-IBE Scheme}

The LW-IBE scheme in prime order bilinear groups is described as follows:

\begin{description}
\item [\textbf{IBE.Setup}($1^\lambda$):] This algorithm first generates the
    asymmetric bilinear groups $\G, \hat{\G}$ of prime order $p$ of bit
    size $\Theta(\lambda)$. It chooses random elements $g \in \G$ and
    $\hat{g}, \hat{w} \in \hat{\G}$. Next, it chooses random exponents
    $\nu, \phi_1, \phi_2 \in \Z_p$ and sets $\tau = \phi_1 + \nu \phi_2$.
    It selects random exponents $\alpha, x, y \in \Z_p$ and sets $u = g^x,
    \hat{u} = \hat{g}^x, h = g^y, \hat{h} = \hat{g}^y, \hat{w}_1 =
    \hat{w}^{\phi_1}, \hat{w}_2 = \hat{w}^{\phi_2}$. It outputs a master
    key $MK = ( \alpha, \hat{g}, \hat{u}, \hat{h}, \hat{w}_1, \hat{w}_2,
    \hat{w})$ and public parameters as
    \begin{align*}
    PP = \Big(~ (p, \G, \hat{\G}, \G_T, e),~
        g, g^{\nu}, g^{-\tau},~ u, u^{\nu}, u^{-\tau},~ h, h^{\nu}, h^{-\tau},~
        \Omega = e(g, \hat{g})^{\alpha}
    ~\Big).
    \end{align*}

\item [\textbf{IBE.GenKey}($ID, MK$):] This algorithm takes as input an
    identity $ID \in \bits^k$ where $k < \lambda$ and the master key $MK$.
    It selects random exponents $r, c_1, c_2 \in \Z_p$ and outputs a
    private key as
    \begin{align*}
    SK_{ID} = \Big(~
    &   K_{1,1} = \hat{g}^{\alpha} (\hat{u}^{ID} \hat{h})^r \hat{w}_1^{c_1},
        K_{1,2} = \hat{w}_2^{c_1},
        K_{1,3} = \hat{w}^{c_1},~
        K_{2,1} = \hat{g}^r \hat{w}_1^{c_2},
        K_{2,2} = \hat{w}_2^{c_2},
        K_{2,3} = \hat{w}^{c_2}
    ~\Big).
    \end{align*}

\item [\textbf{IBE.Encrypt}($M, ID, PP$):] This algorithm takes as input a
    message $M \in \G_T$, an identity $ID$, and the public parameters $PP$.
    It first chooses a random exponent $t \in \Z_p$ and outputs a
    ciphertext as
    \begin{align*}
    CT = \Big(~
    &   C       = e(g, \hat{g})^{\alpha t} M,~
        C_{1,1} = g^t,
        C_{1,2} = (g^{\nu})^t,
        C_{1,3} = (g^{-\tau})^t, \\
    &   C_{2,1} = (u^{ID} h)^t,
        C_{2,2} = ((u^{\nu})^{ID} h^{\nu})^t,
        C_{2,3} = ((u^{-\tau})^{ID} h^{-\tau})^t
    ~\Big).
    \end{align*}

\item [\textbf{IBE.Decrypt}($CT, SK_{ID}, PP$):] This algorithm takes as
    input a ciphertext $CT$, a private key $SK_{ID}$, and the public
    parameters $PP$. If the identities of the ciphertext and the private
    key are equal, then it computes
    \begin{align*}
    M = C \cdot \prod_{i=1}^3 e(C_{1,i}, K_{1,i})^{-1} \cdot
                \prod_{i=1}^3 e(C_{2,i}, K_{2,i}).
    \end{align*}
\end{description}

\subsection{The LW-PKS Scheme}

To derive a LW-PKS scheme from the LW-IBE scheme, we apply the transformation
of Naor \cite{BonehF01}. Additionally, we represent the signature in $\G$
instead of $\hat{\G}$ to reduce the size of signatures. The LW-PKS scheme in
prime order bilinear groups is described as follows:

\begin{description}
\item [\textbf{PKS.KeyGen}($1^\lambda$):] This algorithm first generates
    the asymmetric bilinear groups $\G, \hat{\G}$ of prime order $p$ of bit
    size $\Theta(\lambda)$. It chooses random elements $g, w \in \G$ and
    $\hat{g} \in \hat{\G}$. Next, it chooses random exponents $\nu, \phi_1,
    \phi_2 \in \Z_p$ and sets $\tau = \phi_1 + \nu \phi_2$. It selects
    random exponents $\alpha, x, y \in \Z_p$ and sets $u = g^x, \hat{u} =
    \hat{g}^x, h = g^y, \hat{h} = \hat{g}^y, w_1 = w^{\phi_1}, w_2 =
    w^{\phi_2}$. It outputs a private key $SK = ( \alpha, g, u, h )$ and a
    public key as
    \begin{align*}
    PK = \Big(~ (p, \G, \hat{\G}, \G_T, e),~
        w_1, w_2, w,~
        \hat{g}, \hat{g}^{\nu}, \hat{g}^{-\tau},~
        \hat{u}, \hat{u}^{\nu}, \hat{u}^{-\tau},~
        \hat{h}, \hat{h}^{\nu}, \hat{h}^{-\tau},~
        \Omega = e(g, \hat{g})^{\alpha}
    ~\Big).
    \end{align*}

\item [\textbf{PKS.Sign}($M, SK$):] This algorithm takes as input a message
    $M \in \bits^k$ where $k < \lambda$ and a private key $SK$. It selects
    random exponents $r, c_1, c_2 \in \Z_p$ and outputs a signature as
    \begin{align*}
    \sigma = \Big(~
    &   W_{1,1} = g^{\alpha} (u^M h)^r w_1^{c_1},
        W_{1,2} = w_2^{c_1},
        W_{1,3} = w^{c_1},~
        W_{2,1} = g^r w_1^{c_2},
        W_{2,2} = w_2^{c_2},
        W_{2,3} = w^{c_2}
    ~\Big).
    \end{align*}

\item [\textbf{PKS.Verify}($\sigma, M, PK$):] This algorithm takes as input
    a signature $\sigma$ on a message $M \in \bits^k$ under a public key
    $PK$. It first chooses a random exponent $t \in \Z_p$ and computes
    verification components as
    \begin{align*}
    &   V_{1,1} = \hat{g}^t,
        V_{1,2} = (\hat{g}^{\nu})^t,
        V_{1,3} = (\hat{g}^{-\tau})^t, \\
    &   V_{2,1} = (\hat{u}^M \hat{h})^t,
        V_{2,2} = ((\hat{u}^{\nu})^M \hat{h}^{\nu})^t,
        V_{2,3} = ((\hat{u}^{-\tau})^M \hat{h}^{-\tau})^t.
    \end{align*}
    Next, it verifies that $\prod_{i=1}^3 e(W_{1,i}, V_{1,i}) \cdot
    \prod_{i=1}^3 e(W_{2,i}, V_{2,i})^{-1} \stackrel{?}{=} \Omega^t$. If
    this equation holds, then it outputs $1$. Otherwise, it outputs $0$.
\end{description}

We can safely move the elements $w_1, w_2, w$ from the private key to the
public key since these elements are always constructed in the security proof
of the LW-IBE scheme. However, this LW-PKS scheme does not support multi-user
setting and public re-randomization since the elements $g, u, h$ are not
given in the public key.

\end{document}